\documentclass[12pt,reqno]{amsart}
\usepackage{amsthm,amsfonts,amssymb,euscript}

\newcommand{\bea}{\begin{eqnarray}}
\newcommand{\eea}{\end{eqnarray}}
\def\beaa{\begin{eqnarray*}}
\def\eeaa{\end{eqnarray*}}
\def\ba{\begin{array}}
\def\ea{\end{array}}
\def\be#1{\begin{equation} \label{#1}}
\def \eeq{\end{equation}}

\def\a{{\alpha}}

\def\b{{\beta}}
\def\be{{\beta}}
\def\ga{\gamma}

\def\de{\delta}

\def\ep{\epsilon}
\def\eps{\epsilon}

\def\la{\lambda}

\def\si{\sigma}
\def\Si{\Sigma}

\def\Om{\Omega}

\def\varep{\varepsilon}

\def\pr{{\partial}}
\def\al{\alpha}

\def\rh{{\rho}}

\def\II{{\mathcal I}}

\def\FF{{\mathcal F}}

\def\HH{{\mathcal H}}

\def\SS{{\mathcal S}}

\def\KK{{\mathcal K}}
\def\Lie{{\mathcal L}}

\def\RR{{\mathcal R}}
\def\QQ{{\mathcal Q}}

\def\Lie{{\mathcal L}}

\def\D{{\bf D}}

\def\M{{\bf M}}

\def\O{{\bf O}}

\def\R{{\bf R}}

\def\Z{{\bf Z}}
\def\K{{\bf K}}
\def\T{{\bf T}}
\def\E{{\bf E}}

\def\g{{\bf g}}

\def\f12{{\frac 1 2}}

\def\dual{{\,\,^*}}

\def\lb{{\,\underline{l}}}

\def\Lb{{\,\underline{L}}}

\def\ul{{\underline{l}}}

\def\f{\widetilde{f}}

\newtheorem{theorem}{Theorem}[section]
\newtheorem{lemma}[theorem]{Lemma}
\newtheorem{proposition}[theorem]{Proposition}
\newtheorem{corollary}[theorem]{Corollary}
\newtheorem{definition}[theorem]{Definition}

\numberwithin{equation}{section}

\begin{document}\title[Stationary black holes in vacuum]{Uniqueness of smooth stationary black holes in vacuum: small perturbations of the Kerr spaces}
\author{S. Alexakis}
\address{Massachusetts Institute of Technology}
\email{alexakis@math.mit.edu}
\author{A. D. Ionescu}
\address{University of Wisconsin -- Madison}
\email{ionescu@math.wisc.edu}
\author{S. Klainerman}
\address{Princeton University}
\email{seri@math.princeton.edu}
\thanks{The first author was partially supported by a Clay
research fellowship. The second author was partially supported by
a Packard Fellowship. The third author was partially supported by
NSF grant DMS-0070696.}
\begin{abstract}
Following the program started   in   \cite{IoKl},
 we attempt to remove the ana\-ly\-ti\-city assumption
 in the  the well known Hawking-Carter-Robinson
   uniqueness result for regular stationary  vacuum black holes.
   Unlike  \cite{IoKl},  which was based on a  tensorial  characterization of  the Kerr solutions, due to Mars \cite{Ma1},  we rely  here on Hawking's original  strategy, which is to reduce  the case of  general stationary space-times to that of
   stationary and axi-symmetric  spacetimes for which
    the Carter-Robinson uniqueness result holds. In this
   reduction Hawking had to appeal to analyticity. Using a variant of the  geometric Carleman estimates developed  in \cite{IoKl},  in this paper we show how to  bypass analyticity in the case when the stationary vacuum space-time
   is a small perturbation of a given Kerr solution.  Our  perturbation  assumption is expressed as a uniform smallness condition on the Mars-Simon tensor. The starting point of our proof is  the new   local rigidity theorem established in    \cite{AlIoKl}.
 \end{abstract}
 \maketitle
  \tableofcontents
\section{Introduction}\label{introduction}

It is widely expected\footnote{See reviews
by  B. Carter \cite{Ca-R} and   P. Chusciel \cite{Chrusc-Rev} for a history  and  review of the current status of the conjecture.}
 that the domains of outer communication
of regular, stationary, four dimensional,  vacuum black hole solutions
are isometrically diffeomorphic to those of
the Kerr black holes.  Due to gravitational radiation,
  general, asymptotically flat, dynamic,  solutions of the Einstein-vacuum
  equations  ought to settle down, asymptotically,  into a stationary regime.  Thus  the conjecture, if true, would characterize all possible asymptotic states   of the general  vacuum evolution.  A similar scenario is supposed to hold true in the presence of matter.

So far the  conjecture is known to be  true\footnote{By combining results of Hawking \cite{H-E},
 Carter \cite{Ca1}, and Robinson \cite{Rob}, see also the recent work of Chrusciel-Costa \cite{ChCo}.} if, besides  reasonable geometric and physical  conditions,  one assumes that the space-time metric in the domain of outer communication is \textit{real analytic}.  This last assumption is particularly restrictive,
since there is no reason whatsoever that  general stationary solutions
of the Einstein field equations are  analytic in the  ergoregion, i.e. the region where the stationary Killing vector-field becomes
space-like. Hawking's proof starts with the observation that the event horizon of a general stationary metric  is non-expanding and
the stationary Killing field must be   tangent to it.
 Specializing to the future event horizon  $\HH^+$, Hawking \cite{H-E} (see also \cite{IsMon}) proved the existence
      of a  non-vanishing vector-field $\K$ tangent to the null
generators of $\HH^+$ and  Killing to any order along $\HH^+$.
    Under the assumption of real analyticity
      of the  space-time metric  one can prove, by a Cauchy-Kowalewski type argument
(see \cite{H-E} and  the rigorous   argument in  \cite{Chrusc-1}),  that the Hawking Killing vector-field
$\K$ can be extended to a neighborhood of the entire domain of outer communication. Thus, it follows,  that the spacetime
    $(\M, \g)$ is not just stationary but also
    axi-symmetric. To derive uniqueness, we then
 appeal to the theorem of Carter and Robinson which shows that the exterior region of a
non-degenerate, stationary,
axi-symmetric, connected, connected vacuum black hole must be isometrically diffeomorphic to a Kerr exterior of mass $M$ and angular momentum $a<M$. The proof of this result originally obtained by Carter \cite{Ca1} and Robinson \cite{Rob}, has been strengthened and extended by many authors, notably Mazur \cite{mazur}, Bunting \cite{bunting}, Weinstein \cite{weinst}; the most recent and complete account, which fills in various gaps in the previous literature is the recent paper of Chrusciel and Costa \cite{ChCo}, see also \cite{chrusciel2}. A clear and complete exposition of the ideas that come into the proof can be found in Heusler's book, \cite{heusler}.
 We remark the Carter Robinson theorem does not require analyticity.

    In \cite{IoKl}
     a different strategy was followed  based  on the tensorial characterization of the Kerr spaces,  due to Mars \cite{Ma1} and Simon \cite{Sim},    and  a  new  analytic framework based on Carleman estimates.  Uniqueness of Kerr was proved for a general class of regular stationary vacuum space-times which verify a  complex scalar identity  along the bifurcation sphere of the horizon.  Unfortunately,  to eliminate this local  assumption,  one needs a global argument which  has alluded us so far.

      In this     paper we  return  to Hawking's original strategy and show
      how to extend his Killing vector-field, and thus  axial symmetry, from   the horizon to the entire domain of outer communication without appealing to  analyticity. As noted above, once
once has extended axial symmetry to the entire exterior region,
the Carter-Robinson theorem implies that $({\bf M},{\bf g})$ must be isometric to a Kerr solution.
 Our argument, which relies on the Carleman estimates developed
      in \cite{IoKl} and  \cite{IoKl2}, and their extensions in \cite{Al} and \cite{AlIoKl},
      require  a smallness assumption which is expressed, geometrically, by assuming that the Mars-Simon tensor  of our
stationary    metric is  uniformly bounded by a sufficiently small
constant.  Our main result is
      therefore perturbative; we show that any regular stationary vacuum solution  which
      is sufficiently close to a Kerr solution  $\KK(a, m)$, $0\le a<m$ must  in fact  coincide with it.

      The  first step of our  approach  has already been presented by us
      in \cite{AlIoKl}. There we show, under very general  assumptions,
      how to construct the Hawking Killing vector-field in a neighborhood of a  non-expanding, smooth, bifurcate horizon.
      The main idea, which also plays an essential role in this paper,
      is to turn the problem of extension into one of unique continuation,
      relying on Carleman estimates for systems of wave equations coupled to ordinary differential equations,  see  the introduction in \cite{AlIoKl} for an informal discussion.

       To further extend these vector-fields to the entire domain of outer communication we  make use of the foliation given by the level hypersurfaces of the  function $y$, the real part of  $(1-\si)^{-1}$
      where $\si$ is the complex Ernst potential associated to the
      stationary vector-field $\T$, see subsection \ref{ernst}. The Carleman
      estimates on which our extension argument  is based,  depend
      on a  crucial    \textit{$\T$-conditional pseudo-convexity} property  for the function $y$ (see Lemma \ref{Tpseudo})
      which was  previously shown to hold true (see \cite{IoKl} and
\cite{IoKl2})  if the Mars-Simon tensor $\SS$ vanishes
identically.  Here we show that the the same property holds true
if
      our space-time verifies our small perturbation assumption, i.e.
      $\SS$ sufficiently small. Thus, the main ideas of the paper
      are
      \begin{enumerate}
      \item
      A robust argument by which  the problem of extension of
      Killing vector-fields is turned into a uniqueness problem for an ill-posed system of covariant wave equations      coupled to ODE's.
      \item A local extension argument of Hawking's Killing vector-field in a neighborhood  of the bifurcate
      horizon. This step, which was accomplished in \cite{AlIoKl}, is unconditional, i.e. it
      does not  require the  smallness assumption for $\SS$.
      \item   An extension of the global argument of
      \cite{IoKl}, by which the Hawking vector-field constructed in \cite{AlIoKl} in a neighborhood of the bifurcate horizon can be globally extended. This step, which rests on the $\T$-conditional
       pseudo-convexity property,  requires our global
       smallness assumption for  the Mars-Simon tensor $\SS$.
        \end{enumerate}
      In the subsection below we give precise assumptions and the statement of our main result.

\subsection{Precise assumptions and the main theorem}\label{maintheorem}

We  assume that  $(\M,\g)$ is a smooth\footnote{$\M$ is a connected, oriented, time oriented, paracompact $C^\infty$ manifold without boundary.} vacuum Einstein spacetime of dimension $3+1$ and $\T\in\T(\M)$ is a smooth Killing vector-field on $\M$. We also assume that we are given an embedded partial Cauchy surface $\Sigma^0\subseteq\M$ and a diffeomorphism $\Phi_0:E_{1/2}\to\Sigma^0$, where $E_r=\{x\in\mathbb{R}^3:|x|>r\}$.

We group our main assumptions \footnote{Many of these assumptions
can be justified as consequences of more primitive assumptions,
see \cite{Be-Si2}, \cite{Be-Si}, \cite{Chrusc0}, \cite{CW2},
\cite{CGD}, \cite{FSW}, \cite{FrRaWa}, \cite{Ra-Wa}. For the sake
of simplicity, we do not attempt to work here under the most
general regularity assumptions. See the recent paper \cite{ChCo}
for a careful discussion.} in three categories. The first one
combines a standard  asymptotic flatness assumption with a global
assumption concerning the orbits of $\T$.  The asymptotic flatness
assumption, in particular,  defines the asymptotic region
$\M^{(end)}$ and the domain of outer communication (exterior
region) $ \E=\II^{-}(\M^{(end)})\cap\II^{+}(\M^{(end)}), $ where
$\II^{-}(\M^{(end)})$, $\II^{+}(\M^{(end)})$ denote the past and
respectively  future
 sets of  $\M^{(end)}$. Our  second assumption  concerns
 the smoothness of the two achronal  boundaries $\delta(\II^{-}(\M^{(end)}))$ in a small neighborhood of their intersection $S_0=\de(\II^{-}(\M^{(end)}))\cap \de(\II^{+}(\M^{(end)}))$.
 Our third assumption asserts that $\E$ is a small perturbation a fixed Kerr metric, in a suitable sense. We give an invariant form
 to this assumption by making use of the Mars-Simon tensor  $\SS$ whose vanishing characterizes   Kerr spacetimes,  see \cite{Ma1}.
\medskip

{\bf{GR.}} (Global regularity assumption)
 We assume that  the restriction of the
 diffeomorphism $\Phi_0$ to $E_{R_0}$,  for $R_0$ sufficiently
large,  extends to a diffeomorphism $\Phi_0:\mathbb{R}\times
E_{R_0}\to \M^{(end)}$, where $\M^{(end)}$ (asymptotic region) is
an open subset of $\M$. In local coordinates $\{x^0,x^i\}$ defined
by this diffeomorphism, we assume that $\T=\partial_0$ and, with
$r=\sqrt{(x^1)^2+(x^2)^2+(x^3)^2}$, that the components of the
spacetime metric verify\footnote{\label{foot:asympt}We denote by
 $O_k(r^a)$ any smooth function in  $\M^{(end)}$ which verifies $|\partial^i
f|=O(r^{a-i})$  for any $0\le i\le k$ with $|\partial^i
f|=\sum_{i_0+i_1+i_2+i_3=i}
|\partial_0^{i_0}\partial_1^{i_1}\partial_2^{i_2}\partial_3^{i_3}
f|$.},
\begin{equation}\label{As-Flat}
\g_{00}=-1+\frac{2M}{r}+O_6(r^{-2}),\quad \g_{ij}=\delta_{ij}+O_6(r^{-1}),\quad\g_{0i}=-\ep_{ijk}\frac{2S^jx^k}{r^3}+O_6(r^{-3}),
\end{equation}
for some $M>0$, $S^1,S^2,S^3\in\mathbb{R}$ (see \cite{Be-Si2}) such that,
\begin{equation}\label{As-Flat2}
J=[(S^1)^2+(S^2)^2+(S^3)^2]^{1/2}\in[0,M^2).
\end{equation}

Let
\begin{equation*}
\E=\II^{-}(\M^{(end)})\cap\II^{+}(\M^{(end)}).
\end{equation*}
We assume that $\E$ is globally hyperbolic and
\begin{equation}\label{intersect}
\Sigma^0\cap\II^{-}(\M^{(end)})=\Sigma^0\cap\II^+(\M^{(end)})=\Phi_0(E_1).
\end{equation}
We assume that $\T$ does not vanish at any point of $\E$ and that
every orbit of $\T$ in $\E$ is complete and intersects the
hypersurface $\Sigma^0$.
\medskip

{\bf{SBS.}} (Smooth bifurcation sphere assumption) It follows from \eqref{intersect} that
\begin{equation*}
\delta(\II^{-}(\M^{(end)}))\cap\Sigma^0=\delta(\II^+(\M^{(end)}))\cap\Sigma^0=S_0,
\end{equation*}
where $S_0=\Phi_0(\{x\in\mathbb{R}^3:|x|=1\})$ is an imbedded $2$-sphere (called the bifurcation sphere). We assume that there is a neighborhood $\mathbf{O}$ of $S_0$ in $\mathbf{M}$ such that the sets
\begin{equation*}
\HH^+=\mathbf{O}\cap \delta(\II^{-}(\M^{(end)}))\quad \text{ and }\quad \HH^-=\mathbf{O}\cap \delta(\II^{+}(\M^{(end)}))
\end{equation*}
are smooth imbedded hypersurfaces. We assume that these hypersurfaces are null, non-expanding\footnote{A null hypersurface is said to be non-expanding if the trace of its null second fundamental form vanishes identically.}, and intersect transversally in $S_0$. Finally, we assume that the vector-field $\T$ is tangent to both hypersurfaces $\HH^+$ and $\HH^-$, and does not vanish identically on $S_0$.
\medskip

{\bf{PK.}} (Perturbation of Kerr assumption). Let $\si$ denote the Ernst potential and $\SS$ the Mars-Simon tensor, defined in an open neighborhood of $\Sigma^0\cap\overline{\mathbf{E}}$ in $\mathbf{M}$ (see section \ref{preliminaries} for precise definitions). We assume that
\begin{equation}\label{Main-Cond0}
|(1-\sigma)\SS(\T,T_\al,T_\be,T_\ga)|\leq\overline{\varep}\,\,\text{ on }\Sigma^0\cap\overline{\mathbf{E}},
\end{equation}
for some sufficiently small constant $\overline{\varep}$
(depending only on the constant $\overline{A}$ defined in section
\ref{preliminaries}), where $T_0$ is the future-directed unit
vector orthogonal to $\Sigma^0$ and $T_0,T_1,T_2,T_3$ is an
orthonormal basis along $\Sigma^0$.
\medskip

\noindent {\bf Main Theorem}.\quad \textit{Under the assumptions {\bf{GR}}, {\bf{SBS}}, and {\bf{PK}} the domain of outer communication $\E$ of $\M$ is isometric to the domain of outer communication of the Kerr space-time with mass $M$ and angular momentum $J$.}
\medskip

In other words, a stationary vacuum black hole, which satisfies
suitable regularity assumptions and is sufficiently ``close'' to a
Kerr solution, has to be isometric to that Kerr solution. This can
be interpreted as a strong extension  of Carter's original
theorem, see \cite{Ca1}, \cite{CaL}, on stationary and
axi-symmetric perturbations of the Kerr spaces, in which we remove
the axi-symmetry assumption  and give a geometric, coordinate
independent, perturbation condition. We provide below a more
detailed outline of the proof of the Main Theorem.

In section \ref{preliminaries} we define a system of local coordinates along our reference space-like hypersurface $\Si^0$, and define our main constant $\overline{A}$. The small constant $\overline{\varep}$ in \eqref{Main-Cond0} is to be taken sufficiently small, depending only on $\overline{A}$. We review also the construction of two optical functions $u$ and $\underline{u}$ in a neighborhood of the bifurcation sphere $S_0$, adapted to the null hypersurfaces $\HH^+$ and $\HH^-$, and recall the definition of the complex Ernst potential $\si$ and the Mars-Simon tensor $\SS$. Finally, we record some asymptotic formulas, which are proved in the appendix.

In section \ref{almostKerr} we develop the main consequences of our smallness assumption \eqref{Main-Cond0}. All of our results in this section are summarized in Proposition \ref{clo110}; we prove a lower bound on $|1-\sigma|$ along $\Sigma_1$, as well as several approximate identities in a small neighborhood of $\Sigma_1$ which are used in the rest of the paper.

In section \ref{propertiesy} we derive several properties of the
function $y$ needed in the continuation argument in section
\ref{globalK}. We prove first that $y$ is almost constant on
$S_0$, as in Lemma \ref{yonS}, and increases in a controlled way
in a neighborhood of $S_0$ in $\Sigma_1$. Then we prove that the
level sets of the function $y$ away from $S_0$ are regular, in a
suitable sense. Finally, we prove that the function $y$ satisfies
the $\T$-conditional pseudo-convexity property, away from $S_0$,
see Lemma \ref{Tpseudo}.

In section \ref{globalK}, which is the heart of the paper, we
construct the Hawking Killing vector-field $\K$ in the domain of
outer communication $\E$. The starting point is the existence of
$\K$ in a neighborhood of $S_0$, which was proved in
\cite{AlIoKl}. We extend $\K$ to larger and larger regions, as
measured by the function $y$, as the solution of an ordinary
differential equation, see Lemma \ref{overY}. We then prove that
the resulting vector-field $\K$ is Killing (and satisfies several
other bootstrap conditions) as a consequence of a uniqueness
property of stationary vacuum solutions, see Proposition
\ref{mainnew2}. The proof of this last proposition relies on
Carleman estimates and the properties of the function $y$ proved
in section \ref{propertiesy}.

In section \ref{rotkilling} we construct a global, rotational Killing vector-field $\Z$ which commutes with $\T$, as a linear combination of the vector-fields $\T$ and $\K$. We also
give a simple proof, specialized to our setting, that the span of the two Killing fields $\T,\Z$
 is time-like in $\E$.

\section{Preliminaries}\label{preliminaries}

\subsection{A system of coordinates along $\Sigma^0$} Let $\pr_1,\pr_2,\pr_3$ denote the vectors tangent to $\Sigma^0$, induced by the diffeomorphism $\Phi_0$. Let $\Sigma_r=\Phi_0(E_r)$, where, as before, $E_r=\{x\in\mathbb{R}^3:|x|>r\}$. In particular, for  our original spacelike hypersurface, we have $\Si^0=\Si_{1/2}$.
  Using \eqref{As-Flat} and the assumption that $\Sigma^0$ is spacelike, it follows that there are large constants $A_1$ and $R_1\ge R_0$, such that   $R_1\geq A_1^4$,  with the following properties: on $\Sigma_{3/4}$,
  for any $X=(X^1,X^2,X^3)$,
\begin{equation}\label{prel1}
A_1^{-1}|X|^2\leq\sum_{\al,\be=1}^3X^\al X^\be\g_{\al\be}\leq
A_1|X|^2\quad\text{ and
}\quad\sum_{\al=1}^{3}|\g(\pr_\al,\T)|+|\g(T_0,\T)|\leq A_1.
\end{equation}
In $\Phi_0(\mathbb{R}\times E_{R_1})$, which we continue to denote
by $\M^{(end)}$, $\T=\pr_0$ and  (see notation in footnote
\ref{foot:asympt}),
\begin{equation}\label{prel2}
\begin{split}
\sum_{m=0}^6 r^{m+1}\sum_{j,k=1}^3
|\partial^m(\g_{jk}-\delta_{jk})|&+ \sum_{m=0}^6 r^{m+2}
|\partial^m(\g_{00}+1-2M/r  )|\\
&+\sum_{m=0}^6 r^{m+3} \sum_{i=1}^3
|\partial^m(\g_{0i}+2\ep_{ijk}S^jx^kr^{-3})|\le A_1.
\end{split}
\end{equation}
We construct  a system of coordinates in a small neighborhood $\widetilde{\M}$ of $\Sigma^0\cap\overline{\E}$, which extends both  the
coordinate system of $\M^{(end)}$ in \eqref{prel2} and that of $\Si^0$. We do that with the help of  a smooth vector-field $T'$ which interpolates between $\T$ and $T_0$.
More precisely we  construct  $T'$ in a neighborhood of $\Sigma_{3/4}$ such that $T'=\T$ in $\Phi_0(\mathbb{R}\times E_{2R_1})$ and $T'=\eta(r/R_1)T_0+(1-\eta(r/R_1))\T$ on $\Sigma_{3/4}$, where $\eta:\mathbb{R}\to[0,1]$ is a smooth function supported in $(-\infty,2]$ and equal  to $1$ in $(-\infty,1]$.
  Using now  the flow induced by  $T'$ we extend the original  diffeomorphism $\Phi_0:E_{1/2}\to \Sigma^0$,  to cover a full neighborhood of $\Sigma_1$. Thus  there exists $\varep_0>0$ sufficiently small and a diffeomorphism $\Phi_1:(-\varep_0,\varep_0)\times E_{1-\varep_0}\to\widetilde{\M}$, which agrees with $\Phi_0$ on $\{0\}\times E_{1-\varep_0}\cup (-\varep_0,\varep_0)\times E_{2R_1}$ and  such that $\partial_0=\pr_{x^0}=T'$. By setting $\varep_0$ small enough, we may assume that $\O_{\varep_0}:=\Phi_1((-\varep_0,\varep_0)\times\{x\in\mathbb{R}^3:|x|\in (1-\varep_0,1+\varep_0)\})\subseteq \O$, where $\O$ is the open set defined in the assumption ${\bf{SBS}}$. By construction, using also \eqref{prel2} and letting $\varep_0$ sufficiently small depending on $R_1$,
\begin{equation}\label{prel3}
\sum_{j=1}^3|\g_{0j}|+|\g_{00}+1|\leq A_1/(R_1+r)\quad \text{ in }\widetilde{\M}.
\end{equation}
With $\g_{\al\be}=\g(\pr_\al,\pr_\be)$ and $\T=\T^\al\pr_\al$, let
\begin{equation}\label{prel4}
A_2=\sup_{p\in\widetilde{\M}}\sum_{m=0}^6\Big[\sum_{\al,\be=0}^3|\pr^m\g_{\al\be}(p)|+\sum_{\al=0}^3|\pr^m\T^\al(p)|\Big].
\end{equation}
Finally, we fix
\begin{equation}\label{mainconst}
\overline{A}=\max(R_1,A_2,\varep_0^{-1},(M^2-J)^{-1}).
\end{equation}
The constant $\overline{A}$ is our  main effective constant. The constant $\overline{\varep}$ in \eqref{Main-Cond0} will be fixed sufficiently small, depending only on $\overline{A}$. To summarize, we defined a neighborhood $\widetilde{\M}$ of $\Sigma^0\cap\overline{\E}$ and a diffeomorphism $\Phi_1:(-\varep_0,\varep_0)\times E_{1-\varep_0}\to\widetilde{\M}$, $\varep_0>0$, such that the bounds \eqref{prel1}, \eqref{prel2}, \eqref{prel3}, \eqref{prel4} hold (in coordinates induced by the diffeomorphism $\Phi_1$).

\subsection{Optical functions in a neighborhood of $S_0$} \label{optical}
We define two optical functions $u,\underline{u}$ in a neighborhood of $S_0$. We fix a smooth future-directed  null pair $(L, \Lb)$ along $S_0$, satisfying
\begin{equation}\label{normalization}
\g(L,L)=\g(\Lb,\Lb)=0,\,\,\,\,\,\g(L, \Lb)=-1,
\end{equation}
such that $L$ is tangent to $\HH^+$ and $\Lb$ is tangent to $\HH^-$. In a small neighborhood of $S_0$, we extend $L$ (resp. $\Lb$)  along the null geodesic generators of $\HH^+$
 (resp. $\HH^-$) by parallel transport, i.e. $\D_LL=0$ (resp.  $\D_\Lb\Lb=0$). We define the function $\underline{u}$ (resp. $u$) along $\HH^+$ (resp. $\HH^-$) by setting $u=\underline{u}=0$ on $S_0$ and solving $L(\underline{u})=1$ (resp.  $\Lb(u)=1$).  Let $S_{\underline{u}}$ (resp.  $\underline{S}_{u}$)  be the level surfaces of $\underline{u}$ (resp. $u$)   along $\HH^+$ (resp. $\HH^-$). We define $\Lb$ at every point of $\HH^+$ (resp. $L$ at every point of $\HH^-$) as the unique, future directed null vector-field orthogonal to the surface $S_{\underline{u}}$ (resp. $\underline{S}_u$) passing through that point and such that $\g(L,\Lb)=-1$.
 We now define the null hypersurface $\HH^-_{\underline{u}}$ to be the congruence
 of   null geodesics  initiating on $S_{\underline{u}}\subset\HH^+$ in the direction of $\Lb$.
 Similarly we define $\HH^+_u$ to be the congruence
 of   null geodesics  initiating on $\underline{S}_{u}\subset\HH^-$ in the direction of  $L$.
 Both congruences are well defined in a sufficiently small neighborhood of $S_0$ in $\O$.  The null hypersurfaces $\HH^-_{\underline{u}}$ (resp.  $\HH^+_{u}$) are the level sets
 of  a function $\underline{u}$  (resp $u$)  vanishing on $\HH^-$  (resp. $\HH^+$). By construction
  \begin{equation}\label{haw8}
  L=-\g^{\mu\nu}\pr_\mu u\pr_\nu,\qquad  \Lb=-\g^{\mu\nu}\pr_\mu\underline{u}\pr_\nu.
  \end{equation}
In particular, the functions $u,\underline{u}$ are both null optical functions,
 i.e.
 \begin{equation}\label{haw9}
 \g^{\mu\nu}\pr_\mu u \pr_\nu u=\g(L,L)=0\quad\text{ and }\quad \g^{\mu\nu}\pr_\mu \underline{u}\pr_\nu \underline{u}=\g(\Lb,\Lb)=0.
 \end{equation}

To summarize, there is $c_0=c_0(\overline{A})\in(0,\varep_0]$ sufficiently small and smooth optical functions $u,\underline{u}:\O_{c_0}\to\mathbb{R}$, where $\O_{c_0}=\Phi_1((-c_0,c_0)\times\{x\in\mathbb{R}^3:|x|\in (1-c_0,1+c_0)\})$. In local coordinates induced by the diffeomorphism $\Phi_1$ we have\footnote{Recall the notation $|\partial^jf|=\sum_{j_0+j_1+j_2+j_3=j}|\partial_0^{j_0}\partial_1^{j_1}\partial_2^{j_2}\partial_3^{j_3}f|$, where $\partial_0,\partial_1,\partial_2,\partial_3$ are the derivatives induced by the diffeomorphism $\Phi_1$. This notation will be used throughout the paper.}
\begin{equation}\label{haw10}
\sup_{x\in\O_{c_0}}\sum_{j=0}^4(|\partial^ju(x)|+|\partial^j\underline{u}(x)|)\leq \widetilde{C}=\widetilde{C}(\overline{A}).
\end{equation}
In addition,
\begin{equation}\label{haw11}
\HH^+\cap\O_{c_0}=\{p\in\O_{c_0}:u(p)=0\},\qquad\HH^-\cap\O_{c_0}=\{p\in\O_{c_0}:\underline{u}(p)=0\}.
\end{equation}
In $\O_{c_0}$ we define
 \begin{equation*}
 \Om=\g^{\mu\nu}\pr_\mu u \pr_\nu \underline{u}=\g(L,\Lb).
 \end{equation*}
By construction $\Omega=-1$ on $(\HH^+\cup\HH^-)\cap\O_{c_0}$ (we remark, however, that $\Omega$ is not necessarily equal to $-1$ in $\O_{c_0}$). By taking $c_0$ small enough, we may assume that
\begin{equation}\label{haw12}
\Omega\in[-3/2,-1/2]\quad\text{ in }\O_{c_0}.
\end{equation}
Finally, by construction, we may assume that the functions $|u|,|\underline{u}|$ are proportional to $|1-r|$ on the spacelike hypersurface $\Sigma^0\cap\O_{c_0}$, i.e.
\begin{equation}\label{haw12.5}
|u/(1-r)|, |\underline{u}/(1-r)|\in[\widetilde{C}^{-1},\widetilde{C}]\quad\text{ on }\Sigma^0\cap\O_{c_0},
\end{equation}
where, as in \eqref{haw10}, $\widetilde{C}$ is a constant that depends only on $\overline{A}$.

\subsection{Definitions and asymptotic formulas} \label{ernst} We recall now the definitions of the Ernst potential $\sigma$ and the Mars--Simon tensor $\SS$ (see \cite[Section 4]{IoKl} for a longer discussion and proofs of all of the identities). In $\M$ we define the 2-form,
\beaa
F_{\al\be}=\D_\al\T_\be
\eeaa
and the complex valued 2-form,
\begin{equation}\label{s1}
\FF_{\al\be}=F_{\al\be}+i {\dual F}_{\al\be}=F_{\al\be}+(i/2){\in_{\al\be}}^{\mu\nu}F_{\mu\nu}.
\end{equation}
Let $\FF^2=\FF_{\al\be}\FF^{\al\be}$. We define also the Ernst $1$-form
\begin{equation}\label{s12}
\si_\mu=2\T^\a\FF_{\a\mu}=\D_\mu(-\T^\al \T_\al)- i\in_{\mu\b\ga\de}\T^\b \D^\ga\T^\de.
\end{equation}
It is easy to check that, in $\M$
\begin{equation}\label{Ernst1}
\begin{cases}
&\D_\mu\si_\nu-\D_\nu\si_\mu=0;\\
&\D^\mu\si_\mu=-\FF^2;\\
&\si_\mu\si^\mu=\g(\T,\T) \FF^2.
\end{cases}
\end{equation}
Since $\D_\mu\si_\nu=\D_\nu\si_\mu$ and the sets
$\widetilde{\mathbf{M}}=\Phi_1((-\varep_0,\varep_0)\times
E_{1-\varep_0})$ and $\E$ are simply connected, we can define the
Ernst potential $\si:\widetilde{\mathbf{M}}\cup\E\to\mathbb{C}$
such that $ \si_\mu =\D_\mu\si$, $\Re\si=-\T^\al\T_\al$, and
$\si\to 1$ at infinity along $\Sigma^0$.

We define the complex-valued self-dual Weyl tensor
\begin{equation}\label{s2}
\RR_{\al\be\mu\nu}=R_{\al\be\mu\nu}+(i/2){\in_{\mu\nu}}^{\rh\si}R_{\al\be\rh\si}=R_{\al\be\mu\nu}+i{\dual R}_{\al\be\mu\nu}.
\end{equation}
We define the tensor $\II\in\mathbf{T}_4^0(\mathbf{M})$,
\begin{equation}\label{s4}
\II_{\al\be\mu\nu}=(\g_{\al\mu}\g_{\be\nu}-\g_{\al\nu}\g_{\be\mu}+i\in_{\al\be\mu\nu})/4.
\end{equation}
Let $\widetilde{\M}'=\{p\in\widetilde{\M}:\sigma(p)\neq 1\}$\footnote{Using the assumption {\bf{PK}}, we will prove in section \ref{almostKerr} that $\Sigma_1\subseteq\widetilde{\M}'$.}. We define the tensor-field  $\QQ\in\mathbf{T}_4^0(\widetilde{\mathbf{M}}')$,
\begin{equation}\label{s5}
\QQ_{\al\be\mu\nu}=(1-\sigma)^{-1}\big(\FF_{\al\be}\FF_{\mu\nu}-\frac{1}{3}\FF^2\II_{\al\be\mu\nu}\big).
\end{equation}
It is easy to see that the tensor-field $\QQ$ is a self-dual Weyl field, i.e.
\begin{equation*}
\begin{cases}
&\QQ_{\al\be\mu\nu}=-\QQ_{\be\al\mu\nu}=-\QQ_{\al\be\nu\mu}=\QQ_{\mu\nu\al\be};\\
&\QQ_{\a\b\mu\nu}+\QQ_{\a\mu\nu\be}+\QQ_{\a\nu\be\mu}=0;\\
&\g^{\be\nu}\QQ_{\al\be\mu\nu}=0,
\end{cases}
\end{equation*}
and
\begin{equation*}
^*Q_{\a\b\mu\nu}=\frac{1}{2}\in_{\mu\nu\rh\si}{\QQ_{\al\be}}^{\rh\si}=(-i)\QQ_{\al\be\mu\nu}.
\end{equation*}

We define now the self-dual Weyl field $\SS$, called the Mars--Simon tensor,
\begin{equation}\label{s6}
\SS=\RR+6\QQ.
\end{equation}
We observe that $(1-\sigma)\S$ is a smooth tensor on $\widetilde{\M}$. Using the Ricci identity
\begin{equation}\label{Ricci}
\D_\mu\FF_{\al\be}=\T^\nu\RR_{\nu\mu\al\be},
\end{equation}
and proceeding as in \cite[formula 4.33]{IoKl}, we
deduce that, in  $\widetilde{\M}'$,
\begin{equation}\label{mainiden1}
\D_\rho[\mathcal{F}^2(1-\sigma)^{-4}]=2(1-\sigma)^{-4}\T^\nu \SS_{\nu\rho\ga\de}\FF^{\ga\de}.
\end{equation}
This identity will play a key role in the analysis in section \ref{almostKerr}. Finally, we define the functions $y,z:\widetilde{\M}'\to\mathbb{R}$
\begin{equation*}
y+iz=(1-\sigma)^{-1}.
\end{equation*}

Simple asymptotic computations using the formula \eqref{prel2}, see Appendix \ref{asymptot}, show that, for $R$ sufficiently large depending only on $\overline{A}$,
\begin{equation}\label{asymp0}
\FF_{\al\be}=O(r^{-2}),\quad\SS_{\al\be\ga\de}=O(r^{-3}),\,\,\,\al,\be,\ga,\de=0,\ldots,3.
\end{equation}
More precisely,
\begin{equation}\label{asymp1}
1-\sigma=2Mr^{-1}+O(r^{-2}),\qquad\FF^2=-4M^2r^{-4}+O(r^{-5})
\end{equation}
in $\Phi_1((-\varep_0,\varep_0)\times E_R)$. In particular, $\Phi_1((-\varep_0,\varep_0)\times E_R)\subseteq \widetilde{\M}'$ and
\begin{equation}\label{asymp2}
-4M^2\FF^2(1-\sigma)^{-4}=1+O(r^{-1})\qquad\text{ in }\quad\Phi_1((-\varep_0,\varep_0)\times E_R).
\end{equation}
In addition
\begin{equation}\label{asymp3}
y=\frac{r}{2M}+O(1),\qquad z=\frac{S^1x^1+S^2x^2+S^3 x^3}{2M^2r}+O(r^{-1})
\end{equation}
in $\Phi_1((-\varep_0,\varep_0)\times E_R)$. Finally,
\begin{equation}\label{asymp4}
z^2+4M^2(y^2+z^2)\D_\mu z\D^\mu z=\frac{J^2}{4M^4}+O(r^{-1})\qquad\text{ in }\quad\Phi_1((-\varep_0,\varep_0)\times E_R).
\end{equation}
All these asymptotic identities are proved in Appendix \ref{asymptot} and will be used in section \ref{almostKerr}.

\section{Analysis on the hypersurface $\Sigma_1$}\label{almostKerr}

In this section we use assumption {\bf{PK}} to prove several
approximate identities on the hypersurface
$\Sigma_1=\Sigma^0\cap\E=\Phi_0(E_1)$. The general idea is to
prove approximate identities such as \eqref{asymp2} and
\eqref{asymp4} first in the asymptotic region, using the
asymptotic flatness assumption \eqref{prel2}, and then extend them
to the entire hypersurface $\Sigma_1$ using the fact that the
Mars--Simon tensor is assumed to be small. We will prove also that
$1-\sigma$ does not vanish in $\Sigma_1$. All of our results in
this section are summarized in Proposition \ref{clo110}.

We will use the notation in section \ref{preliminaries}. We fix first a large constant $\overline{R}$ which depends only on our main constant $\overline{A}$ (see \eqref{mainconst}). We fix $r_0$ the smallest number in $[1,\overline{R}-1]$ with the property that
\begin{equation}\label{clo1}
|1-\sigma|\geq \overline{R}^{-2}\quad\text{ on }\Sigma_{r_0}\setminus\Sigma_{\overline{R}},
\end{equation}
where, as before, $\Sigma_r=\Phi_0(E_r)$, $E_r=\{x\in\R^3:|x|>r\}$. Such an $r_0$ exists if $\overline{R}$ is sufficiently large, in view of \eqref{asymp1} and the continuity of $1-\sigma$. We will prove, among other things, that $r_0=1$. In this section we let $\widetilde{C}$ denote various constants in $[1,\infty)$ that may depend only on $\overline{R}$ (thus on $\overline{A}$ once $\overline{R}$ is fixed sufficiently large depending on $\overline{A}$). The value of $\overline{\varep}$ in \eqref{Main-Cond0} is assumed to be sufficiently small depending on the constants $\widetilde{C}$. To summarize, $\log(\overline{A})\ll \log(\overline{R})\ll \log(\widetilde{C})\ll\log (\overline{\varep}^{-1})$.

We will work in the region $\Phi_1[(-\overline{\varep},\overline{\varep})\times E_{r_0}]$. Since $|\partial_0(1-\sigma)|\leq \widetilde{C}$, it follows from \eqref{clo1} and \eqref{asymp1} that
\begin{equation}\label{clo2}
|1-\sigma|^{-1}\leq \widetilde{C}r\quad\text{ in }\Phi_1[(-\overline{\varep},\overline{\varep})\times E_{r_0}].
\end{equation}
Using the assumption \eqref{Main-Cond0} and the asymptotic identities \eqref{asymp0}, we have
\begin{equation}\label{clo0}
|(1-\sigma)\T^\nu\SS_{\nu\rho\ga\de}|\leq \widetilde{C}\min(\overline{\varep},r^{-4})\quad\text{ in }\Phi_1[(-\overline{\varep},\overline{\varep})\times E_{r_0}],
\end{equation}
in the coordinate frame $\pr_0,\pr_1,\pr_2,\pr_3$. Using \eqref{mainiden1}, \eqref{asymp0} and the last two inequalities, it follows that
\begin{equation}\label{clo3}
|\partial_\rho(\FF^2(1-\sigma)^{-4})|\leq\widetilde{C}r^3\min(\overline{\varep},r^{-4}) \quad\text{ in }\Phi_1[(-\overline{\varep},\overline{\varep})\times E_{r_0}].
\end{equation}

We prove now that
\begin{equation}\label{clo3.5}
|1+4M^2\FF^2(1-\sigma)^{-4}|\leq\widetilde{C}\min(r^{-1},\overline{\varep}^{1/5})\quad\text{ in }\Phi_1[(-\overline{\varep},\overline{\varep})\times E_{r_0}].
\end{equation}
Indeed, let $H=1+4M^2\FF^2(1-\sigma)^{-4}$. Using \eqref{asymp2} and \eqref{clo3},
\begin{equation}\label{clo3.6}
|H|\leq\widetilde{C}r^{-1}\quad\text{ and }\quad\sum_{\rho=0}^4|\partial_\rho H|\leq \widetilde{C}r^3\min(\overline{\varep},r^{-4}) \quad\text{ in }\Phi_1[(-\overline{\varep},\overline{\varep})\times E_{r_0}].
\end{equation}
The bound \eqref{clo3.5} follows from the first inequality in \eqref{clo3.6} at points $p$ for which $r(p)\geq\varep^{-1/5}$. To prove \eqref{clo3.5} at points $p$ with $r(p)\leq\varep^{-1/5}$ we fix a point $p'\in\Phi_1[(-\overline{\varep},\overline{\varep})\times E_{r_0}]$ with $r(p')=\varep^{-1/5}$. We integrate along a line joining the points $p$ and $p'$ and use the second inequality in \eqref{clo3.6}. The result is $|H(p)-H(p')|\leq \widetilde{C}\varep^{-4/5}\overline{\varep}$, which gives \eqref{clo3.5} since $|H(p')|\leq\widetilde{C}r(p')^{-1}=\widetilde{C}\varep^{-1/5}$.

It follows from \eqref{clo3} and \eqref{clo3.5} that there is a smooth function $G_1:\Phi_1[(-\overline{\varep},\overline{\varep})\times E_{r_0}]\to\mathbb{C}$ with the properties
\begin{equation}\label{big1}
\begin{split}
-4M^2\FF^2=(1-\sigma)^4(1+G_1)^2,\qquad|G_1|+\sum_{\rho=0}^3|\pr_\rho G_1|\leq \widetilde{C}\min(r^{-1},\overline{\varep}^{1/5})
\end{split}
\end{equation}
on $\Phi_1[(-\overline{\varep},\overline{\varep})\times E_{r_0}]$. In particular, using also \eqref{clo2},
\begin{equation}\label{clo7}
|\FF^2|\geq (\widetilde{C}r)^{-4}\quad\text{ in }\Phi_1[(-\overline{\varep},\overline{\varep})\times E_{r_0}].
\end{equation}

We define the smooth function $P=y+iz:\Phi_1[(-\overline{\varep},\overline{\varep})\times E_{r_0}]\to\mathbb{C}$,
\begin{equation}\label{se90}
P=y+iz=(1-\sigma)^{-1}.
\end{equation}
We construct now a special   null pair, similar to the principal
null pair in  \cite[Section 4]{Ma1}.

\begin{lemma}\label{nullpair}
 There  exists a  future-directed  null pair $l,\ul$,  $\g(l,\lb)=-1$,  such that
\begin{equation}\label{big2}
\FF_{\al\be}l^\be=(1+G_1)(4MP^2)^{-1}l_\a,\qquad \FF_{\al\be}\ul^\b=-(1+G_1)(4MP^2)^{-1}\ul_\a,
\end{equation}
in $\Phi_1[(-\overline{\varep},\overline{\varep})\times E_{r_0}]$.
\end{lemma}
\begin{proof}[Proof of Lemma \ref{nullpair}] Let $Z_\a$ be  complex eigenvector  $\FF_{\a\b}Z^\b=\la Z_\a$ with complex
eigenvalue $\la$. Using the relation  $\FF_{\a\si}\FF_{\b}\,^\si=(1/4)\g_{\a\b}\FF^2$
 (see  \cite[formula 4.2]{IoKl})
we derive,
\begin{equation*}
\la^2=-\frac{1}{4}\FF^2=\frac{1}{16M^2}(1-\si)^4(1+G_1)^2.
\end{equation*}
Thus, $\la=\pm  (4M P^2)^{-1}(1+G_1)$.  The reality
of the corresponding eigenvectors $l,\lb$  is a consequence of  the self duality of $\FF$. They must both be null in view of the antisymmetry of $\FF$ and can be normalized apropriately.
\end{proof}

Let $e_{(3)}=\ul$, $e_{(4)}=l$. We fix vector-fields $e_{(1)},e_{(2)}$ in $\Phi_1[(-\overline{\varep},\overline{\varep})\times E_{r_0}]$ such that together with $e_{(3)}=\ul$, $e_{(4)}=l$ they form a  positively oriented  null frame, i.e.,
\begin{equation}\label{clo20}
\begin{split}
&\g(l,e_{(1)})=\g(l,e_{(2)})=\g(\ul,e_{(1)})=\g(\ul,e_{(2)})=\g(e_{(1)},e_{(2)})=0,\\
&\g(e_{(1)},e_{(1)})=\g(e_{(2)},e_{(2)})=\in_{(1)(2)(3)(4)}=1.
\end{split}
\end{equation}
In view of \eqref{clo7}, the vector-fields $e_{(\mu)}=e_{(\mu)}^\al\partial_\al$, $\mu=1,2,3,4$ can be chosen such that
\begin{equation}\label{clo21}
\sum_{\mu=1}^4\sum_{\al=0}^3|e_{(\mu)}^\al|\leq\widetilde{C}\quad\text{ in }\Phi_1[(-\overline{\varep}.\overline{\varep})\times E_{r_0}].
\end{equation}
According to \eqref{big2},  \eqref{clo20} and the self duality of $\FF$, the  components of $\FF$ are,
\begin{equation}\label{se6'}
\FF_{(4)(1)}=\FF_{(4)(2)}=\FF_{(3)(1)}=\FF_{(3)(2)}=0\text{ and }\FF_{(4)(3)}=i\FF_{(2)(1)}=(1+G_1)/(4MP^2).
\end{equation}
 This is equivalent to the identity,
\begin{equation}\label{se5}
\FF_{\a\b}=\frac{1+G_1}{4MP^2}\big(-l_\a\ul_\b+l_\b\ul_\a-i\in_{\al\be\mu\nu}l^\mu\ul^\nu\big).
\end{equation}
By contracting \eqref{se5} with $2\T^\a$ and using $2\T^\a\FF_{\a\b}=\sigma_\be=\D_\be\sigma$ we derive
\begin{equation*}
\D_\be (y+iz)=\frac{1+G_1}{2M}\big[-(\T^\al l_\al)\ul_\b+(\T^\al \ul_\al)l_\be-i\in_{\al\be\mu\nu}\T^\a l^\mu\ul^\nu\big].
\end{equation*}
In particular, if $G_1=\Re G_1+i\Im G_1$, we have
\begin{equation}\label{good3}
\begin{split}
&\D_\be y=\frac{1+\Re G_1}{2M}\big[-(\T^\al l_\al)\ul_\b+(\T^\al \ul_\al)l_\be\big]+\frac{\Im G_1}{2M}\in_{\al\be\mu\nu}\T^\a l^\mu\ul^\nu,\\
&\D_\be z=-\frac{1+\Re G_1}{2M}\in_{\al\be\mu\nu}\T^\a l^\mu\ul^\nu+\frac{\Im G_1}{2M}\big[-(\T^\al l_\al)\ul_\b+(\T^\al \ul_\al)l_\be\big].
\end{split}
\end{equation}
It follows from \eqref{good3} and \eqref{big1} that
\begin{equation}\label{vgood3}
|\D_{(1)}y|+|\D_{(2)}y|+|\D_{(3)}z|+|\D_{(4)}z|\leq \widetilde{C}\min(r^{-1},\overline{\varep}^{1/5})\quad\text{ in }\Phi_1[(-\overline{\varep},\overline{\varep})\times E_{r_0}].
\end{equation}

A direct computation using the definition of $P$ and \eqref{Ernst1} shows that
\begin{equation}\label{he9}
\D_\a P\D^\a P=\frac{\D_\a\sigma \D^\a\sigma}{(1-\sigma)^4}=-\frac{(1+G_1)^2\T^\al\T_\al}{4M^2}.
\end{equation}
Since $-\T^\al\T_\al=\Re\sigma=1-y/(y^2+z^2)$ we have
\begin{equation}\label{se10}
\begin{split}
&\D_\a y\D^\a y-\D_\a z\D^\a z=\frac{(1+\Re G_1)^2-(\Im G_1)^2}{4M^2}\Big(1-\frac{y}{y^2+z^2}\Big),\\
&\D_\a y\D^\a z=\frac{(1+\Re G_1)\Im G_1}{4M^2}\Big(1-\frac{y}{y^2+z^2}\Big).
\end{split}
\end{equation}

\subsection{A lemma of Mars}   The following  lemma is an adaptation to our situation
of  an important  calculation which first appears in  \cite{Ma1}.
\begin{lemma}\label{clo100}
With $B=J^2/(4M^4)<1/4$ we have in $\Phi_1[(-\overline{\varep},\overline{\varep})\times E_{r_0}]$,
\begin{equation}\label{conclz}
|4M^2(y^2+z^2)\D_\be z\D^\be z+z^2-B|\leq\widetilde{C}\min(r^{-1},\overline{\varep}^{1/40}).
\end{equation}
\end{lemma}
\begin{proof}[Proof of Lemma \ref{clo100}]
We show that the  function $H:=4M^2(y^2+z^2)\D_\be z\D^\be z+z^2$ is almost constant by computing its derivatives with respect
to our null frame \eqref{clo20}.       In the particular case $\SS=0$, the constancy  of $H$ was first proved
 in \cite{Ma1}  using  the  full  Newman-Penrose
formalism. A similar proof was  later  given in  \cite{IoKl}. Here
we give instead a straightforward proof based only on the
formulas we have derived so far.

We will prove that
\begin{equation}\label{derbound}
\sum_{\al=0}^3|\partial_\al H(p)|\leq \widetilde{C}\overline{\varep}^{1/20}\qquad\text{ if }p\in\Phi_1[(-\overline{\varep},\overline{\varep})\times E_{r_0}]\text{ and }r(p)\leq \varep^{-1/40}.
\end{equation}
Assuming this, the bound \eqref{conclz} follows from the bound $|H-B|\leq \widetilde{C}r^{-1}$ in $\Phi_1[(-\overline{\varep},\overline{\varep})\times E_{r_0}]$, see \eqref{asymp4}, in the same way the bound \eqref{clo3.5} follows from \eqref{clo3} and \eqref{asymp2}.

We differentiate $H$ and derive,
\begin{equation}\label{eq:DH}
\D_\al H=8M^2(y^2+z^2)\D_\al\D_\be z\D^\be z+8M^2(y\D_\al y+z\D_\al z)\D_\be z\D^\be z+2z\D_\al z.
\end{equation}
To calculate the main term  $8M^2(y^2+z^2)\D_\al\D_\be z\D^\be z$
we  first   calculate  the second covariant derivatives of $P=(y+iz)$,
using the definition of $\SS$  and \eqref{Ricci},
\begin{equation*}
\begin{split}
\D_\al\D_\be P&=2(1-\sigma)^{-3}\D_\al\si\D_\be\si+(1-\sigma)^{-2}\D_\al\D_\be\si\\
&=2(1-\sigma)^{-3}\D_\al\si\D_\be\si+2(1-\sigma)^{-2}({F_\al}^\rho\FF_{\rho\be}+\T^\rho\T^\nu\RR_{\nu\al\rho\be})\\
&=2(1-\si)^{-3}\D_\al\si\D_\be\si+2(1-\sigma)^{-2}{F_\al}^\rho\FF_{\rho\be}\\
&-12(1-\si)^{-2}\T^\rho\T^\nu\QQ_{\nu\al\rho\be}+2(1-\si)^{-2}\T^\rho\T^\nu\SS_{\nu\al\rho\be}\\
&=2(1-\sigma)^{-2}{F_\al}^\rho\FF_{\rho\be}+2(1-\si)^{-2}\T^\rho\T^\nu\SS_{\nu\al\rho\be}\\
&-(1-\si)^{-3}\si_\al\si_\be+(1-\si)^{-3}\FF^2[\g_{\al\be}(\T^\rho\T_\rho)-\T_\al\T_\be]\\
&=2P^2{F_\al}^\rho\FF_{\rho\be}+P^{-1}[(\D_\rho P\D^\rho P)\g_{\al\be}-\D_\al P\D_\be P]\\
&+2P^2\T^\rho\T^\nu\SS_{\nu\al\rho\be}-P^3\FF^2\T_\al\T_\be.
\end{split}
\end{equation*}
Thus, we have the identity
\begin{equation}\label{he2}
\begin{split}
\D_\al\D_\be P=&-P^{-1}\D_\al P\D_\be P+P^{-1}(\D_\rho P\D^\rho P)\g_{\al\be}\\
&-2P^2{F_\al}^\rho\FF_{\be\rho}-P^3\FF^2\T_\al\T_\be+2P^2\T^\rho\T^\nu\SS_{\nu\al\rho\be}.
\end{split}
\end{equation}
Since $\T(z)=0$ we deduce,
\begin{equation}\label{he3}
\begin{split}
\D_\al\D_\be P \D^\b z&=P^{-1}(\D_\rho P\D^\rho P)\D_\a z-2P^2 {F_\al}^\rho\FF_{\be\rho}\D^\b z  \\
& - P^{-1}\D_\al P\D_\be P \D^\b z+2P^2\T^\rho\T^\nu\SS_{\nu\al\rho\be}\D^\b z.
\end{split}
\end{equation}
Observe that  $\D_\al\D_\be z\D^\be z= \Im\big[\D_\a\D_\b P \D^\b z\big]$. Thus, in view of \eqref{he3}
\begin{equation}\label{eq:interm}
\begin{split}
\D_\al\D_\be z\D^\be z
&=\Im\big[-2P^2
{F_\al}^\rho\FF_{\be\rho}\,\D^\b z+2P^2\T^\rho\T^\nu\SS_{\nu\al\rho\be}\D^\b z\big]\\
&-\Im\big[P^{-1}\D_\al P\D_\be P \D^\b z\big]+\Im\big[P^{-1}(\D_\rho P\D^\rho P)\D_\a z\big].
\end{split}
\end{equation}
Now,
\beaa
(y^2+z^2)\Im\big[P^{-1}\D_\al P\D_\be   P\,  \D^\b z \big]
&=& \D^{\be} z\Im[(y-iz)\D_{\a} (y+iz)\D_{\b}(y+iz)]\\
&=& (y\D_\a y +z\D_\a z) \D^\b z \D_\b z+(y \D_\a z-z\D_\a y )\D^\b y \D_\b z
\eeaa
and,
\beaa
(y^2+z^2)\Im\big[P^{-1}\D_\rho P\D^\rho    P\,  \D_\a z \big]&=&\D_\a z\Im[
[(y-iz)\D_{\rho} (y+iz)\D^{\rho}(y+iz)]\\
&=&2y \D_\a z\D_\rho y \D^\rho z -z \D_\a z(\D_\rho y \D^\rho y-\D_\rho z \D^\rho z).
\eeaa
Therefore, back to \eqref{eq:interm},
\begin{equation*}
\begin{split}
(y^2+z^2)\D_\al\D_\be z\D^\be z&=
(y^2+z^2)\Im\big[ -2P^2
{F_\al}^\rho\FF_{\be\rho}\,\D^\b z+2P^2\T^\rho\T^\nu\SS_{\nu\al\rho\be}\D^\b z\big]\\
&-y \D_\a y \, \D_\rho z\D^\rho z+(y\D_\a z+z\D_\a y) \, \D_\rho y\D^\rho z- z\D_\a z \, \D_\rho y\D^\rho y.
\end{split}
\end{equation*}
Going  back to \eqref{eq:DH} we derive,
\begin{equation*}
\begin{split}
\D_\al H&=8 M^2 (y^2+z^2)\Im\big[ -2P^2
{F_\al}^\rho\FF_{\be\rho}\,\D^\b z+2P^2\T^\rho\T^\nu\SS_{\nu\al\rho\be}\D^\b z\big]\\
&+8M^2 z\D_\al z(\D_\rho z\D^\rho z- \D_\rho y\D^\rho y)+        2z\D_\al z+8M^2(y\D_\a z+z\D_\a y) \, \D_\rho y\D^\rho z.
\end{split}
   \end{equation*}

Recall that we are looking to prove \eqref{derbound} at points $p$ with $r(p)\leq \overline{\varep}^{-1/40}$.
 In view of \eqref{clo2} and \eqref{clo0}, at such points we have
 \beaa
 16M^2(y^2+z^2)\D^{\be} zP^2\T^{\rho}\T^{\nu}\SS_{\nu\a\rho\b}&=&O_1(\overline{\varep}),
 \eeaa
where, for simplicity of notation, in this lemma we let $O_1(\overline{\varep})$ denote any quantity bounded by $\widetilde{C}\overline{\varep}^{1/20}$. According to \eqref{big1} and \eqref{se10} we also have,
 \begin{equation*}
|\D_\rho y \D^\rho z|+\Big|\D_\rho z\D^\rho z- \D_\rho y\D^\rho y+\frac{1}{4M^2}\big(1-\frac{y}{y^2+z^2}\big)\Big|\leq\widetilde{C}\overline{\varep}^{1/5}.
 \end{equation*}
 Thus, using again \eqref{clo2},
 \beaa
 8M^2 z \D_\a z\big(\D_\rho z\D^\rho z- \D_\rho y\D^\rho y\big) + 2z\D_a z&=&\frac{2yz\D_\a z}{y^2+z^2} +O_1(\overline{\varep}).
 \eeaa
 Consequently,
 \begin{equation*}
\begin{split}
\D_\al H&=-16 M^2 (y^2+z^2)\Im\big[ P^2
{F_\al}^\rho\FF_{\be\rho}\,\D^\b z]   +\frac{2yz\D_\a z}{y^2+z^2} +O_1(\overline{\varep}).
   \end{split}
   \end{equation*}
For \eqref{derbound} it only  remains to check  that,
   \bea
   -16 M^2 (y^2+z^2)\Im\big[ P^2
{F_\al}^\rho\FF_{\be\rho}\,\D^\b z]   +\frac{2yz\D_\a z}{y^2+z^2}=O_1(\overline{\varep}). \label{eq:FFDz}
   \eea
We prove this in the null frame $e_{(1)}, e_{(2)}, e_{(3)}, e_{(4)}$, see \eqref{clo20}. Recalling \eqref{se6'}  and \eqref{vgood3} we easily see that
    both terms on the left are bounded by $\widetilde{C}\overline{\varep}^{1/20}$ for $\a=3, 4$.  For  $\a=1$, using \eqref{se6'} and \eqref{big1},
\begin{equation*}
\begin{split}
-16 M^2 (y^2+z^2)\Im\big[ &P^2{F_{(1)}}^{(\rho)}\FF_{(\be)(\rho)}\,\D^{(\b)} z]=-16 M^2 (y^2+z^2)\Im\big[ P^2F_{(1)(2)}\FF_{(1)(2)}\,\D_{(1)} z]\\
&=-16 M^2 (y^2+z^2)\D_{(1)} z\Im\Big[P^2\frac{i}{4MP^2}\frac{2yz}{4M(y^2+z^2)^2}\Big]+O_1(\overline{\varep}).
\end{split}
\end{equation*}
The approximate identity \eqref{eq:FFDz} follows for $\al=1$. The proof of \eqref{eq:FFDz} for $\al=2$ is similar, which completes the proof of the lemma.
\end{proof}

\subsection{Conclusions}\label{sec:pseudo.conv}
It follows from  \eqref{se10} and Lemma \ref{clo100},  that
\begin{equation}\label{concly}
\D_\be z \D^\be z=\frac{B-z^2}{4M^2(y^2+z^2)}+O(\overline{\varep}),\quad\D_\be y \D^\be y=\frac{y^2-y+B}{4M^2(y^2+z^2)}+O(\overline{\varep})
\end{equation}
in $\Phi_1[(-\overline{\varep},\overline{\varep})\times E_{r_0}]$, where $O(\overline{\varep})$ denotes functions on $\Phi_1[(-\overline{\varep},\overline{\varep})\times E_{r_0}]$ dominated by $\widetilde{C}\min(r^{-1},\overline{\varep}^{1/40})$. Using \eqref{good3}  we deduce that
$\D^\b y\D_\b y=\frac{1}{2M^2}(\T^\al l_\al)(\T^\be\ul_\be)+O(\overline{\varep})$. Hence,
\begin{equation}\label{conclt}
(\T^\al l_\al)(\T^\be\ul_\be)=\frac{y^2-y+B}{2(y^2+z^2)}+O(\overline{\varep})\quad\text{ in }\Phi_1[(-\overline{\varep},\overline{\varep})\times E_{r_0}].
\end{equation}

We prove now that the value of $r_0$ in \eqref{clo1} can be taken to be equal to $1$. It view of the definition of $r_0$, it suffices to prove the following lemma:

\begin{lemma}\label{clo101}
Assuming $\overline{R}$ is chosen sufficiently large, we have
\begin{equation*}
|1-\sigma|\geq 2\overline{R}^{-2}\quad\text{ on }\Sigma_{r_0}\setminus\Sigma_{\overline{R}}.
\end{equation*}
\end{lemma}

\begin{proof}[Proof of  Lemma \ref{clo101}] The conclusion of the lemma is equivalent to
\begin{equation}\label{clo50}
|P|\leq \overline{R}^2/2\quad\text{ on }\Sigma_{r_0}\setminus\Sigma_{\overline{R}}.
\end{equation}
To prove this, we recall that we still have the flexibility to fix $\overline{R}$ sufficiently large depending on $\overline{A}$. In view of \eqref{asymp1}, for \eqref{clo50} it suffices to prove that
\begin{equation}\label{clo51}
|\pr_\al P|\leq C(\overline{A})\quad\text{ on }\Sigma_{r_0},
\end{equation}
for $\al=1,2,3$ and some constant $C(\overline{A})$ that depends only on $\overline{A}$. The bound \eqref{clo51} follows from \eqref{asymp0} and \eqref{asymp1} at points $p$ for which $r(p)\geq R(\overline{A})$. Since $P=(1-\sigma)^{-1}$ and $|1-\sigma|\geq |\Re(1-\sigma)|=|1+\T^\al\T_\a|$, the bound \eqref{clo51} also follows at points $p$ for which $r(p)\leq R(\overline{A})$ and $\g_p(\T,\T)\notin[-3/2,-1/2]$.

It remains to prove the bound \eqref{clo51} at points $p\in\Sigma_{r_0}$ for which $r(p)\leq R(\overline{A})$ and $\g_p(\T,\T)\in[-3/2,-1/2]$. Since $|1-\sigma|\leq C(\overline{A})$ on $\Sigma_1$, we have $|y|+|z|\geq C(\overline{A})^{-1}$ on $\Sigma_1$. It follows from \eqref{concly} that
\begin{equation}\label{clo52}
|\D_\be z\D^\be z|+|\D_\be y\D^\be y|\leq C(\overline{A})\quad\text{ on }\Sigma_{r_0}.
\end{equation}
In addition, $\T(\sigma)=0$ therefore $\T^\al\D_\al z=\T^\al\D_\al y=0$. Since $\g_p(\T,\T)\in[-3/2,-1/2]$ it follows that $\T_p$ is timelike, thus the vectors $Y^\al=\D^\al y$ and $Z^\al=\D^\al z$ are spacelike at the point $p$. The elliptic bounds \eqref{prel1} show, in fact, that $\sum_{\be=1}^3|\pr_\al y|^2\leq C(\overline{A})|\D_\be y\D^\be y|$ and $\sum_{\be=1}^3|\pr_\al z|^2\leq C(\overline{A})|\D_\be z\D^\be z|$, so  \eqref{clo51} follows from \eqref{clo52}.
\end{proof}

We summarize the main conclusions of our analysis so far in the following proposition:

\begin{proposition}\label{clo110}
There is a constant $\widetilde{C}=\widetilde{C}(\overline{A})$ sufficiently large such that
\begin{equation*}
|1-\sigma|\geq(\widetilde{C}r)^{-1}\quad\text{ on }\Sigma_{1-\overline{\varep}},
\end{equation*}
provided that $\overline{\varep}$ is sufficiently small (depending on $\overline{A}$). Therefore the frame $e_{(\al)}$, $\al=1,\ldots,4$, the Mars--Simon tensor $\SS$, and the functions $P,y,z,G_1$ are well defined in $\Phi_1[(-\overline{\varep},\overline{\varep})\times E_{1-\overline{\varep}}]$. In addition, the identities and inequalities \eqref{clo0}, \eqref{big1}, \eqref{clo21}, \eqref{se6'}, \eqref{good3}, \eqref{vgood3}, \eqref{se10}, \eqref{he2}, \eqref{concly}, \eqref{conclt} hold in $\Phi_1[(-\overline{\varep},\overline{\varep})\times E_{1-\overline{\varep}}]$.
\end{proposition}

In view of the assumption ${\bf{GR}}$ on the orbits of $\T$, it
follows that the functions $y=\Re[(1-\sigma)^{-1}]$ and
$z=\Im[(1-\sigma)^{-1}]$ are well defined smooth functions on
$\E$.

\section{Properties of the function $y$}\label{propertiesy}

Our next goal is to understand the behaviour of the function $y$ defined in \eqref{se90} on $\Sigma_1$. Most of our analysis in the next section depends on having sufficiently good information on $y$, both in a small neighborhood of the bifurcation sphere $S_0$ and away from this small neighborhood.\footnote{For comparison $y=r/(2M)$, in the Kerr space of mass $M$ and angular momentum $J$, in standard Boyer--Lindquist coordinates.} In this section we use the notation $\widetilde{C}$ to denote various constants in $[1,\infty)$ that may depend only on the main constant $\overline{A}$. We assume implicitly that $\overline{\varep}^{-1}$ is sufficiently large compared to all such constants $\widetilde{C}$.

\subsection{Control of $y$ in a neighborhood of $S_0$}

We analyze first the function of $y$ in a neighborhood of the bifurcation sphere $S_0$.

\begin{lemma}\label{yonS}
On the bifurcation sphere $S_0$,
\begin{equation}\label{clo70}
|y-(1+\sqrt{1-4B})/2|+\sum_{\al=0}^3|\pr_\al y|\leq\widetilde{C}\overline{\varep}^{1/40}.
\end{equation}
Moreover, there are constants $r_1=r_1(\overline{A})>1$ and $\widetilde{C}_1=\widetilde{C}_1(\overline{A})\gg 1$ such that
\begin{equation}\label{clo70.5}
\widetilde{C}_1(r-1)^2+\widetilde{C}_1\overline{\varep}^{1/40}\geq y-(1+\sqrt{1-4B})/2\geq\widetilde{C}_1^{-1}(r-1)^2-\widetilde{C}_1\overline{\varep}^{1/40}\quad\text{ on }\Sigma_{1}\setminus\Sigma_{r_1}.
\end{equation}
\end{lemma}

\begin{proof}[Proof of Lemma \ref{yonS}] Recall the vector-fields $L$, $\Lb$ defined in a neighborhood of $S_0$ in section \ref{preliminaries}. It is easy to prove, see for example \cite[Section 5]{IoKl}, that
\begin{equation*}
\FF_{\al\be}L^\be=\FF(L,\Lb)L_\al\quad\text{ and  }\quad \FF_{\al\be}\Lb^\be=\FF(\Lb,L)\Lb_\al\quad\text{ on }S_0.
\end{equation*}
Since the vectors $l$ and $\ul$ constructed in Lemma \ref{nullpair} are the unique solutions of the systems of equations $(\FF_{\al\be}\pm f\g_{\al\be})V^\be=0$, up to  rescaling and relabeling, we may assume that $L=l$ and $\Lb=\ul$ on $S_0$. Thus, we may also assume that $e_{(1)}, e_{(2)}$ are tangent to $S_0$.  Since $\T$ is tangent to $S_0$,
\begin{equation*}
L(\sigma)=L^\be\si_\be=2L^\be\T^\al\FF_{\al\be}=0\quad\text{ on }S_0.
\end{equation*}
Similarly, $\Lb(\sigma)=0$ on $S_0$. Using also \eqref{vgood3}, we conclude that
\begin{equation}\label{clo71}
e_{(3)}(y)=e_{(4)}(y)=0\quad\text{ and }\quad|e_{(1)}(y)|+|e_{(2)}(y)|\leq\widetilde{C}\overline{\varep}^{1/5}\quad\text{ on }S_0.
\end{equation}

The inequality on the gradient of $y$ in \eqref{clo70} follows from \eqref{clo71}. For the remaining inequality we use first \eqref{concly}. It follows from \eqref{clo71} that $|\D_\be y\D^\be y|\leq\widetilde{C}\overline{\varep}^{2/5}$ on $S_0$, thus
\begin{equation*}
|y^2-y+B|\leq\widetilde{C}\overline{\varep}^{1/40}\quad\text{ on }S_0.
\end{equation*}
Since $B=J^2/(4M^4)\in[0,1/4)$ (see \eqref{As-Flat2} and \eqref{mainconst}), it follows that
\begin{equation}\label{clo72}
|y-(1+\sqrt{1-4B})/2|\leq\widetilde{C}\overline{\varep}^{1/40}\,\,\text{ on }S_0\,\,\text{ or }\,\,|y-(1-\sqrt{1-4B})/2|\leq\widetilde{C}\overline{\varep}^{1/40}\,\,\text{ on }S_0.
\end{equation}

To eliminate the second alternative we start by deriving
 a wave equation for $y$. Since $\D^\mu\D_\mu\si=-\FF^2$, $\D^\mu\si\D_\mu\si=-\FF^2\Re\si$ (see \eqref{Ernst1}), and $-\FF^2=(1-\sigma)^4(1+G_1)^2/(4M^2)$ we derive
\begin{equation*}
\begin{split}
\D^\mu\D_\mu P&=(1-\sigma)^{-2}\D^\mu\D_\mu\sigma +2(1-\sigma)^{-3}\D^\mu\si\D_\mu\si\\
&=\frac{(1+G_1)^2}{4M^2}(1-\si)(1+\overline{\si})=\frac{2\overline{P}-1}{4M^2P\overline{P}}(1+G_1)^2.
\end{split}
\end{equation*}
Thus, using \eqref{big1}
\begin{equation}\label{clo80}
\D^\mu\D_\mu y=\frac{2y-1}{4M^2(y^2+z^2)}+E,\quad|E|\leq \widetilde{C}\overline{\varep}^{1/5},\quad\text{ on }\Sigma_{1-\overline{\varep}}.
\end{equation}

We  now compare $y$ with a function $y'$ which coincides
with $y$ on $\HH^+$  and verifies $\Lb(y)=0$.
We  use the notation in section \ref{preliminaries}. For $\varep_1=\varep_1(\overline{A})\in(0,c_0]$ sufficiently small we define the function
\begin{equation}\label{clo81}
y':\O_{\varep_1}\to\mathbb{R},\quad y'=y\,\text{ on }\,\HH^+\cap\O_{\varep_1},\quad\Lb(y')=0\text{ in }\O_{\varep_1}.
\end{equation}
The functions $y$ and $y'$ are smooth on $\O_{\varep_1}$, and, using \eqref{clo71} and the definition of $y'$
\begin{equation}\label{clo82}
y-y'=0\quad\text{ on }(\HH^+\cup\HH^-)\cap\O_{\varep_1}.
\end{equation}
In addition,  using again \eqref{clo71},  we infer that
 $|e_{(\a)}(y')|\leq\widetilde{C}\overline{\varep}^{1/5}$ on $S_0$, for all  $\a=1,2,3,4$.  Using \eqref{good3} $e_{(a)}(y')=(2M)^{-1}\Im G_1\in_{(\al)(a)(\mu)(\nu)}\T^\al l^\mu\ul^\nu$ on $S_0$, $a=1,2$. It follows from \eqref{big1} and the inequality $|e_{(\a)}(y')|\leq\widetilde{C}\overline{\varep}^{1/5}$ that $|\D_{(a)}\D_{(a)}y'|\leq\widetilde{C}\overline{\varep}$ on $S_0$, $a=1,2$. Using $\Lb(y')=0$ and $|e_{(\a)}(y')|\leq\widetilde{C}\overline{\varep}^{1/5}$, we have $|\D_{(3)}\D_{(4)}y'|+|\D_{(4)}\D_{(3)}y'|\leq\widetilde{C}\overline{\varep}$ on $S_0$. Therefore,
\begin{equation}\label{clo83}
|\square_\g y'|\leq\widetilde{C}\overline{\varep}^{1/5}\quad\text{ on }S_0.
\end{equation}
 In view of \eqref{clo82}, there is $\varep_2=\varep_2(\overline{A})\in(0,\varep_1)$ such that $y-y'=u\underline{u}f$ in $\O_{\varep_2}$, where $f:\O_{\varep_2}\to\mathbb{R}$ is smooth. Since $u=\underline{u}=0$ and $2\D^\al u\D_\al\underline{u}=-2$ on $S_0$, it follows from \eqref{clo80} and \eqref{clo83} that
\begin{equation}\label{clo84}
f=-(1/2)\D^\al\D_\al(y-y')=\frac{1-2y}{8M^2(y^2+z^2)}+E',\quad|E'|\leq \widetilde{C}\overline{\varep}^{1/5},\quad\text{ on }S_0.
\end{equation}
To summarize, $y=y'+u\underline{u}f$ in $\O_{\varep_2}$, where $f$ satisfies \eqref{clo84} on $S_0$.

We eliminate now the second alternative in \eqref{clo72}. The main
point is that if $y$ is close to $(1-\sqrt{1-4B})/2<1/2$ on $S_0$
then $f$ is strictly positive on $S_0$ (see \eqref{clo84}). In
quantitative terms, there is
$\varep_3=\varep_3(\overline{A})\in(0,\varep_2)$ such that
$f\geq\widetilde{C}^{-1}$ in $\O_{\varep_3}$. Since $u\underline
u\leq 0$ on $\Sigma_1\cap\O_{\varep_3}$ and $y'\leq
(1-\sqrt{1-4B})/2+\widetilde{C}\overline{\varep}^{1/40}$ on
$\O_{\varep_3}$ (using the second alternative in \eqref{clo72} and
the construction of $y'$), it follows that
\begin{equation*}
y\leq (1-\sqrt{1-4B})/2-\widetilde{C}^{-1}|u\underline{u}|+\widetilde{C}\overline{\varep}^{1/40}\quad\text{ on }\Sigma_1\cap\O_{\varep_3}.
\end{equation*}
In particular, using \eqref{haw12.5}, $y\leq
(1-\sqrt{1-4B})/2-\widetilde{C}^{-1}$ at some point in $\Sigma_1$
(provided, of course, that $\overline{\varep}$ is sufficiently
small depending on $\overline{A}$). Using also \eqref{asymp3}, it
follows that function $y:\Sigma_1\to\mathbb{R}$ attains its
minimum at some point $p\in\Sigma_1$, and $y(p)\leq
(1-\sqrt{1-4B})/2-\widetilde{C}^{-1}$. Thus
$T_1(y)=T_2(y)=T_3(y)=0$ at $p$ (where $T_0,T_1,T_2,T_3$ is the
orthonormal frame along $\Sigma_1$ defined in assumption
{\bf{PK}}), and
\begin{equation*}
\D_\al y\D^\al y=-(T_0(y))^2\leq 0
\end{equation*}
at p. This is in contradiction, however, with the identity
\eqref{concly} and the inequality $y(p)\leq
(1-\sqrt{1-4B})/2-\widetilde{C}^{-1}$. We conclude that the first
alternative in \eqref{clo70} holds.

We prove now the second statement of the lemma: since $y$ is close
to $(1+\sqrt{1-4B})/2>1/2$ on $S_0$, it follows from \eqref{clo84}
there is $\varep_3=\varep_3(\overline{A})\in(0,\varep_2)$ such
that $f\in[-\widetilde{C},-\widetilde{C}^{-1}]$ in
$\O_{\varep_3}$. Also $|y'-(1+\sqrt{1-4B})/2|\leq
\widetilde{C}\overline{\varep}^{1/40}$ in $\O_{\varep_3}$ and
$u\underline{u}/(r-1)^2\in [-\widetilde{C},-\widetilde{C}^{-1}]$
(see \eqref{haw12.5}). The inequalities in \eqref{clo70.5} follow
since $y=y'+fu\underline{u}$.
\end{proof}

\subsection{Regularity properties of the function $y$ away from $S_0$}
We derive now the main properties of the level sets of the function $y$. Recall first  the main inequality proved in Lemma \ref{yonS}: there are constants $r_1=r_1(\overline{A})>1$ and $\widetilde{C}_1=\widetilde{C}_1(\overline{A})\gg 1$ such that
\begin{equation}\label{extrecall}
\widetilde{C}_1(r-1)^2+\widetilde{C}_1\overline{\varep}^{1/40}\geq y-(1+\sqrt{1-4B})/2\geq\widetilde{C}_1^{-1}(r-1)^2-\widetilde{C}_1\overline{\varep}^{1/40}\quad\text{ on }\Sigma_{1}\setminus\Sigma_{r_1}.
\end{equation}
We define
\begin{equation}\label{ext4}
y_0=(1+\sqrt{1-4B})/2+\widetilde{C}_2^{-1},
\end{equation}
where we fix $\widetilde{C}_2=\widetilde{C}_2(\overline{A})$ a
sufficiently large constant depending on the constants
$\widetilde{C}_1$, $r_1$ in \eqref{extrecall} and $\overline{c}$
in Proposition \ref{summary}. As in \cite[Section 8]{IoKl}, for
$R\in[y_0,\infty)$ we define
\begin{equation}\label{ext5}
\begin{split}
&\mathcal{V}_R=\{p\in\Sigma_1:y(p)<R\};\\
&\mathcal{U}_R=\text{ the connected component of }\mathcal{V}_R\text{ whose closure in }\Sigma^0\text{ contains }S_0.
\end{split}
\end{equation}
In view of \eqref{extrecall},
\begin{equation*}
\Sigma_1\setminus\Sigma_{1+(4\widetilde{C}_1\widetilde{C}_2)^{-1/2}}\subseteq\mathcal{V}_{y_0}\cap (\Sigma_1\setminus\Sigma_{r_1})\subseteq\mathcal{V}_{y_0+\widetilde{C}_2^{-1}}\cap (\Sigma_1\setminus\Sigma_{r_1})\subseteq\Sigma_1\setminus\Sigma_{1+(4\widetilde{C}_1\widetilde{C}_2^{-1})^{1/2}},
\end{equation*}
provided that $\widetilde{C}_2$ is sufficiently large and $\overline{\varep}$ is sufficiently small.     In particular,  we deduce,
\begin{equation}\label{ext5.5}
\Sigma_1\setminus\Sigma_{1+(4\widetilde{C}_1\widetilde{C}_2)^{-1/2}}\subseteq\mathcal{U}_{y_0}\subseteq\mathcal{U}_{y_0+\widetilde{C}_2^{-1}}\subseteq\Sigma_1\setminus\Sigma_{1+(4\widetilde{C}_1\widetilde{C}_2^{-1})^{1/2}}.
\end{equation}

For $p=\Phi_1(0,q)\in\Sigma_1$ and $r\leq\varep_0$  we define,
\begin{equation*}
B_r(p)=\Phi_1\big(\{(t,q')\in(-\varep_0,\varep_0)\times E_{1-\varep_0}:t^2+|q-q'|^2<r^2\}\big).
\end{equation*}
For any set $U\subseteq\Sigma_1$ let $\delta_{\Sigma_1}(U)$ denote its boundary in $\Sigma_1$. Clearly, if $p\in\delta_{\Sigma_1}(\mathcal{U}_R)$ for some $R\geq y_0$ then $y(p)=R$.

We define the vector-field $Y=\D^\al y\D_\al$ in
$\Phi_1[(-\overline{\varep},\overline{\varep})\times
E_{1-\overline{\varep}}]$ and its projection $Y'$ along  the
hypersurface $\Sigma_{1-\overline{\varep}}$,
\begin{equation}\label{ext9}
Y'=Y+\g(Y,T_0)T_0=\sum_{\al=1}^{3}(Y')^\al\pr_\al.
\end{equation}
The vector-field $Y'$ is smooth, tangent to the hypersurface
$\Sigma_{1-\overline{\varep}}$, and
\begin{equation*}
\sum_{\al=1}^3|(Y')^\al|\leq\widetilde{C}\delta_0^{-1}\quad\text{
on }\Sigma_{1-\varep}.
\end{equation*}
In addition,
\begin{equation}\label{alexnew1}
Y'(y)=\g(Y',Y)=\g(Y,Y)+\g(Y,T_0)^2\geq \g(Y,Y)=\D_\al y\D^\al y.
\end{equation}
In particular, if $p\in\delta_{\Sigma_1}(\mathcal{U}_R)$ for some
$R\geq y_0$ then $y(p)=R$ thus, using \eqref{concly},
$Y'(y)(p)\geq \widetilde{C}^{-1}$. Therefore if
$p\in\delta_{\Sigma_1}(\mathcal{U}_R)$ then
\begin{equation}\label{ext015}
\{x\in
B_{\delta}(p)\cap\Sigma_1:y(x)<R\}=B_{\delta}(p)\cap\mathcal{U}_R,
\end{equation}
for any $\delta\leq\delta_1=\delta_1(\overline{A})>0$.

We prove now that the regions $\mathcal{U}_R$, $R\geq y_0$,
increase in a controlled way.

\begin{lemma}\label{Uincrease}
There is $\delta_2=\delta_2(\overline{A})>\in(0,\delta_1)$ such
that, for any $\delta\leq \delta_2$ and $R\in[y_0,\infty)$,
\begin{equation}\label{ext7}
\cup_{p\in \mathcal{U}_R}(B_{\delta^3}(p)\cap\Sigma_1)\subseteq \mathcal{U}_{R+\delta^2}\subseteq\cup_{p\in \mathcal{U}_R}(B_{\delta}(p)\cap\Sigma_1).
\end{equation}
In addition
\begin{equation}\label{ext8}
\cup_{R\geq y_0}\mathcal{U}_R=\Sigma_1
\end{equation}
and
\begin{equation}\label{alexnew2}
\mathcal{U}_R=\mathcal{V}_R\qquad\text{ for any }R\geq y_0.
\end{equation}
\end{lemma}

\begin{proof}[Proof of Lemma \ref{Uincrease}] The first inclusion in \eqref{ext7} is clear: since $y$ is a smooth function
in a neighborhood of $\overline{\Sigma_1}$ (see Proposition
\ref{clo110}), it follows that $y(q)<R+\delta^2$ for any
$p\in\mathcal{U}_R$ and $q\in B_{\delta^3}(p)\cap\Sigma_1$,
provided that $\delta$ is sufficiently small.

To prove the second inclusion, it suffices to prove that
\begin{equation}\label{ext20}
\mathcal{U}_{R+\delta^2}\subseteq\mathcal{U}_R\cup[\cup_{p\in \delta_{\Sigma_1}(\mathcal{U}_R)}(B_{\delta/4}(p)\cap\Sigma_1)],
\end{equation}
for $\delta$ sufficiently small, $R\ge y_0$. Assume, for contradiction, that $q$ is a point in $\mathcal{U}_{R+\delta^2}$ which does not belong to the open set in the right-hand side of \eqref{ext20}. Let $\gamma:[0,1]\to\mathcal{U}_{R+\delta^2}\cup S_0$ be a continuous curve such that $\gamma(0)\in S_0$ and $\gamma(1)=q$ (see definition \eqref{ext5}). Let $q'=\gamma(t')$, $t'\in(0,1]$, denote the first point on this curve which does not belong to the open set in the right-hand side of \eqref{ext20}. Clearly, $q'$ does not belong to the closure of $\mathcal{U}_R$ in $\Sigma_1$, thus $q'$ belongs to the closure of the set $\cup_{p\in \delta_{\Sigma_1}(\mathcal{U}_R)}(B_{\delta/4}(p)\cap\Sigma_1)$ in $\Sigma_1$. Since $\delta_{\Sigma_1}(\mathcal{U}_R)$ is a compact set (see \eqref{asymp3} and \eqref{ext5.5}), it follows that
\begin{equation}\label{ext21}
q'\in B_{\delta/2}(p_0)\cap\Sigma_1\quad\text{ for some }p_0\in\delta_{\Sigma_1}(\mathcal{U}_R).
\end{equation}

For $p\in\Sigma_{1+(4\widetilde{C}_1\widetilde{C}_2)^{-1/2}}$ and
$|t|\leq\delta'$, $\delta'>0$ sufficiently small, let
$\gamma_p(t)\subseteq \Sigma_{1}$ denote the integral curves of
the vector-field $Y'$ defined in \eqref{ext9}, starting at $p$.
Using \eqref{alexnew1}, the fact that $y(p_0)=R\geq y_0$, and
\eqref{concly}, it follows that
\begin{equation}\label{ext10}
Y'(y)\geq\widetilde{C}^{-1}\quad\text{ in
}B_\delta(p_0)\cap\Sigma_1,
\end{equation}
provided that $\delta$ is sufficiently small. With $q'\in
B_{\delta/2}(p_0)$ being the point constructed earlier, we look at
the curve $\gamma_{q'}(t)$, $t\in[-\delta^{3/2},\delta^{3/2}]$.
Clearly, this curve is included in $B_\delta(p_0)\cap\Sigma_1$,
assuming $\delta$ sufficiently  small. Using \eqref{ext10} and the
fact that $y(q')<R+\delta^2$ (since
$q'\in\mathcal{U}_{R+\delta^2}$), we derive that there is a point
$q''$ on the curve $\gamma_{q'}(t)$,
$t\in[-\delta^{3/2},\delta^{3/2}]$, such that $y(q'')<R$. It
follows from \eqref{ext015} that $q''\in\mathcal{U}_R$. Since
$q'\notin\mathcal{U}_R$ (by construction), there is a point
$q'''=\gamma_{q'}(t''')$, $t'''\in[-\delta^{3/2},\delta^{3/2}]$,
such that $q'''\in\delta_{\Sigma_1}(\mathcal{U}_R)$. It follows
that $q'\in B_{\delta/8}(q''')$, in contradiction with the fact
that $q'$ does not belong to the set in the right-hand side of
\eqref{ext20}. This completes the proof of \eqref{ext20}.

The completeness property  \eqref{ext8} follows easily from the
asymptotic formula  \eqref{asymp3} and the fact that $y$ is a
smooth function on $\Sigma_1$.

To prove \eqref{alexnew2} we notice that, in view of
\eqref{ext5.5}, it suffices to prove that $\mathcal{V}_R\cap
\Sigma_{1+(4\widetilde{C}_1\widetilde{C}_2)^{-1/2}}\subseteq
\mathcal{U}_R$, for any $R\geq y_0$. Assume, for contradiction,
that
\begin{equation}\label{ext30}
\text{ there is }R_0>y_0\text{ and }q\in
\Sigma_{1+(4\widetilde{C}_1\widetilde{C}_2)^{-1/2}}\text{ such
that }y(q)<R_0\text{ and }q\notin\mathcal{U}_{R_0}.
\end{equation}
Let $I=\{R\in[R_0,\infty):q\notin\mathcal{U}_R\}$. Since  $I$ is
bounded, due to \eqref{ext8}, we can take $R'$ its least upper
bound. We analyze two possibilities: $q\in \mathcal{U}_{R'}$ and
$q\notin\mathcal{U}_{R'}$.

If $q\in\mathcal{U}_{R'}$ then, using \eqref{ext30}, $R'>R_0$. For
$\delta>0$ sufficiently small (depending on $R'-R_0$ and
$\overline{A}$), it follows from \eqref{ext7} that there is
$R''=R'-\delta^2\geq R_0+\delta^{1/2}$ and a point $q'\in
\mathcal{U}_{R''}$ such that $|q-q'|<\delta$. However, $y(q)<R_0$,
see \eqref{ext30}, thus $y(x)<R_0+\delta^{1/2}\leq R''$ for any
$x\in B_\delta(q)$. Since
$B_\delta(q)\cap\mathcal{U}_{R''}\neq\emptyset$, it follows that
$q\in \mathcal{U}_{R''}$, in contradiction with the definition of
$R'$.

Finally, assume that $q\notin\mathcal{U}_{R'}$. Then
$q\notin\delta_{\Sigma_1}(\mathcal{U}_{R'})$,  in view of
\eqref{ext015} and \eqref{ext30}. For $\delta$ sufficiently small
(smaller than the distance between $q$ and the compact set
$\delta_{\Sigma_1}(\mathcal{U}_{R'})$), it follows from
\eqref{ext20} that $q\notin\mathcal{U}_{R'+\delta^2}$. This is in
contradiction with the definition of $R'$, which completes the
proof of \eqref{alexnew2}.
\end{proof}

\subsection{$\T$-conditional pseudo-convexity} In this subsection we prove a $\T$-conditional pseudo-convexity
property of the function $y$ away from the bifurcation sphere
$S_0$. This pseudo-convexity property, which was first observed in
\cite{IoKl2} in the case of the Kerr spaces and  used in
\cite{IoKl}, plays a key role in the Carleman estimates and the
uniqueness arguments in the next section. We remark that the main
condition for pseudo-convexity is the assumption \eqref{mainneed}
below.

Since $\mathcal{U}_{y_0}=\{p\in\Sigma_1:y(p)<y_0\}$, see
\eqref{alexnew2} and the definition \eqref{ext5}, it follows from
\eqref{ext5.5} that $y\geq (1+\sqrt{1-4B})/2+\widetilde{C}_2^{-1}$
in $\Sigma_{1+(4\widetilde{C}_1\widetilde{C}_2^{-1})^{1/2}}$.
Using also \eqref{extrecall}, it follows that
\begin{equation*}
y\geq (1+\sqrt{1-4B})/2+\widetilde{C}_2^{-2}\quad\text{ in
}\Sigma_{1+(8\widetilde{C}_1\widetilde{C}_2)^{-1/2}}.
\end{equation*}

\begin{lemma}\label{Tpseudo}
Assume $p\in\Sigma_{1+(8\widetilde{C}_1\widetilde{C}_2)^{-1/2}}$,
thus
\begin{equation}\label{mainneed}
y(p)\geq(1+\sqrt{1-4B})/2+\widetilde{C}_2^{-2}.
\end{equation}
There there is a constant $c_2=c_2(\overline{A})>0$ and $\mu=\mu(p)\in\mathbb{R}$ such that
\begin{equation}\label{mainconcl}
X^\al X^\be(\mu\g_{\al\be}(p)-\D_\al\D_\be y(p))\geq c_2|X|^2
\end{equation}
for any real vector $X$ with the property that
\begin{equation}\label{X1}
|X^\al\T_\al(p)|+|X^\al\D_\al y(p)|\leq c_2|X|.
\end{equation}
\end{lemma}
\begin{proof}[Proof of Lemma \ref{Tpseudo}] The bound \eqref{mainconcl} follows easily with $\mu=1$ if $r(p)\geq\widetilde{C}$ is sufficiently large (see \eqref{prel3} and \eqref{asymp3}). Assume that $r(p)\leq\widetilde{C}$. We shall make use of the null frame $e_{(\al)}$ defined in section \ref{almostKerr}.

Using  \eqref{he2}, we write
\beaa
X^\a X^\b\D_\al\D_\be y &=&X^\a X^\b\, \Re[ -P^{-1}\D_\al P\D_\be P]+\g(X, X)
\, \Re[P^{-1}(\D_\rho P\D^\rho P)]\\
&-&2    X^\a X^\b \,  \Re[P^2{F_\al}^\rho\FF_{\be\rho}]-(X^\al\T_\al)^2\Re(P^3\FF^2)+|X|^2O(\overline{\varep}),
\eeaa
where, in this proof, $O(\overline{\varep})$ denotes quantities bounded by $\widetilde{C}\overline{\varep}^{1/40}$. Since $X(y)=|X|O(c_2)$,
\bea
X^\a X^\b\Re[ -P^{-1}\D_\al P\D_\be P]&=&\frac{y}{y^2+z^2} X(z)^2+|X|^2O(c_2),\label{eq:crucial1}
\eea
where, in this proof, $O(c_2)$ denotes quantities bounded by $\widetilde{C}c_2$. Using   \eqref{he9}  and \eqref{big1},
\beaa
\Re[P^{-1}(\D_\rho P\D^\rho P)]&=& \frac{y}{y^2+z^2}\frac{1}{4M^2}\big(1-\frac{y}{y^2+z^2}\big)+O(\overline{\varep}),
\eeaa
thus
\begin{equation}\label{eq:crucial2}
X^\a X^\b \Re[P^{-1}(\D_\rho P\D^\rho P)]\g_{\a\b}=\g(X,X) \frac{y(y^2+z^2-y)}{4M^2(y^2+z^2)^2} +|X|^2O(\overline{\varep}).
\end{equation}
To calculate $X^\a X^\b\Re[\,P^2\, F_{\a}\,^\rho\FF_{\b\rho}] $ we recall, see \eqref{se6'} that
  all components of $\FF$ vanish, with the exception of
  $\FF_{34}=-\FF_{43}=- \frac{1}{4M P^2} +O(\overline{\varep})$ and    $\FF_{12}=-\FF_{21}=\frac{i}{4M P^2}+O(\overline{\varep})$, $a, b=1,2$.
  Since $F=\Re(\FF)$  we also  have,
\begin{equation*}
  F_{34}=-F_{43}=-(4M)^{-1}\Re[P^{-2}]+O(\overline{\varep}),\quad F_{12}=-F_{21}=-(4M)^{-1}\Im[P^{-2}]+O(\overline{\varep}).
\end{equation*}
  Therefore,
\beaa
X^\a X^\b  F_{\a}\,^\rho\FF_{\b\rho}&=& 2 X^3  X^4 F_{34}\, \FF_{34}+\big((X^1)^2+(X^2)^2\big)F_{12}\FF_{12}+|X|^2O(\overline{\varep}) \\
&=&\frac{1}{16M^2P^2}\bigg(2 X^3 X^4 \Re[P^{-2}]-i\big( (X^1)^2+(X^2)^2\big)\Im[P^{-2}]
\,\bigg)+|X|^2O(\overline{\varep}).
\eeaa

Thus,
\bea
-2    X^\a X^\b\,   \Re[P^2{F_\al}^\rho\FF_{\be\rho}]&=&-X^3 X^4 \frac{y^2-z^2}{4M^2(y^2+z^2)^2}+|X|^2O(\overline{\varep}).
\eea
Therefore, denoting $E(X,X):=X^\al X^\be(\mu\g_{\al\be}-\D_\al\D_\be y)$, we write
\beaa
E(X,X)
 &=&\g(X,X) \Big(\mu -\frac{y(y^2+z^2-y)}{4M^2(y^2+z^2)^2} \Big) -\frac{y}{y^2+z^2} X(z)^2\\
 &+&X^3 X^4 \frac{y^2-z^2}{4M^2(y^2+z^2)^2}+|X|^2O(\overline{\varep})+|X|^2O(c_2),
  \eeaa
  or, since $\g(X,X) =-2 X^3 X^4+ (X^1)^2+(X^2)^2$,
  \beaa
  E(X,X)&=&2 X^3 X^4\big[\frac{y(y^2+z^2-y)}{4M^2(y^2+z^2)^2}+
  \frac{y^2-z^2}{8M^2(y^2+z^2)^2}-\mu\big] -\frac{y}{y^2+z^2} X(z)^2\\
  &+&\big((X^1)^2+(X^2)^2\big) \big(\mu -\frac{y(y^2+z^2-y)}{4M^2(y^2+z^2)^2} \big)    +|X|^2O(\overline{\varep})+|X|^2O(c_2)\\
  &=&2 X^3 X^4\big(\frac{2y-1}{8M^2(y^2+z^2)}-\mu\big) -\frac{y}{y^2+z^2} X(z)^2\\
  &+&\big((X^1)^2+(X^2)^2\big) \big(\mu -\frac{y(y^2+z^2-y)}{4M^2(y^2+z^2)^2} \big)    +|X|^2O(\overline{\varep})+|X|^2O(c_2).
  \eeaa

We now make use of our main identity \eqref{conclz}  as well as
\eqref{vgood3}, and derive, \beaa (\D_1 z)^2 +(\D_2 z)^2=\D_\be z
\D^\b
z+O(\overline{\varep})=\frac{B-z^2}{4M^2(y^2+z^2)}+O_1(\overline{\varep}).
\eeaa Thus, using also $\D_3z=O(\overline{\varep})$ and
$\D_4z=O(\overline{\varep})$, by Cauchy-Schwartz,
\begin{equation*}
\begin{split}
X(z)^2&\le \big( (X^1)^2+(X^2)^2\big) \big( (\D_1 z)^2 +(\D_2 z)^2\big)+O(\overline{\varep})\\
&\le
 \big( (X^1)^2+(X^2)^2\big)\frac{B-z^2}{4M^2(y^2+z^2)}+O(\overline{\varep}).
\end{split}
\end{equation*}
We deduce,
\beaa
E(X,X)&\ge&|X|^2O(\overline{\varep})+ |X|^2O(c_2)+2 X^3 X^4\bigg[\frac{2y-1}{8M^2(y^2+z^2)}-\mu\bigg] \\
&+&\quad \big( (X^1)^2+(X^2)^2\big)\,\bigg[\, \mu -\frac{y(y^2+z^2-y)}{4M^2(y^2+z^2)^2}    - \frac{y(B-z^2)}{4M^2(y^2+z^2)^2} \bigg],
\eeaa
or,
\begin{equation}\label{eq:main.ineq}
\begin{split}
 E(X,X)\ge|X|^2O(\overline{\varep})&+ |X|^2O(c_2)+ 2X^3 X^4\bigg[\frac{2y-1}{8M^2(y^2+z^2)}-\mu\bigg]\\
&+\big( (X^1)^2+(X^2)^2\big)\,\bigg[\, \mu -\frac{y(y^2-y+B)}{4M^2(y^2+z^2)^2}
  \bigg].
\end{split}
\end{equation}
 Since $X(y)=X^\a\D_\a y=|X|O(c_2)$, it follows from \eqref{good3} that $X^{4}(\T^{\a} l_{\a})-X^{3}(\T^{\a}\ul_{\a})=|X|O(\overline{\varep})+|X|O(c_2)$. On the other hand, according to \eqref{conclt},
 \beaa
(\T^\al l_\al)(\T^\be\ul_\be)=\frac{y^2-y+B}{2(y^2+z^2)}+O(\overline{\varep})
\eeaa
and therefore, in view of \eqref{mainneed} and $B<1/4$, $(\T^\al l_\al)(\T^\be\ul_\be)\ge c'=c'(\overline{A})>0$ (recall $r(p)\leq\widetilde{C}$). We infer that,
\begin{equation*}
2X^{3}X^{4}\geq\widetilde{C}^{-1}[(X^{3})^2+(X^{4})]^2-\widetilde{C}\overline{\varep}|X|-\widetilde{C}c_2|X|.
\end{equation*}
Thus the expression in \eqref{eq:main.ineq} is bounded from below by $c_2|X|^2$ if $c_2$ is sufficiently small and the coefficients of $((X^{1})^2+(X^{2})^2)$ and $2X^{3}X^{4}$ are both positive, for a suitable choice of $\mu$.  This holds if and only if,
\beaa
\frac{y(y^2-y+B)}{4M^2(y^2+z^2)^2}<\mu<\frac{2y-1}{8M^2(y^2+z^2)}.
\eeaa
Such a choice exists since,
\beaa
\frac{2y-1}{8M^2(y^2+z^2)}-\frac{y(y^2-y+B)}{4M^2(y^2+z^2)^2}=
\frac{ z^2(2y-1)+y(y-2B)   }{8M^2(y^2+z^2)^2}\ge \widetilde{C}^{-1}.
\eeaa
This last inequality holds because $B<1/4$, $y\geq 1/2+c_1$, and $r(p)\leq \widetilde{C}$.
\end{proof}

\section{Construction of the Hawking Killing vector-field $\K$}\label{globalK}
In this section we construct a second Killing vector-field $\K$ in
$\E\cup\O_{\overline{c}}$, for some small constant
$\overline{c}=\overline{c}(\overline{A})\in(0,c_0)$. The first
step, the existence of $\K$ in a neighborhood of $S_0$, was proved
by the authors in \cite{AlIoKl}. We summarize first the main
results in \cite{AlIoKl}, see Theorem 1.1, Proposition 4.5,
Proposition 5.1, and Proposition 5.2, in a suitable quantitative
form.

\begin{proposition}\label{summary}
There is a neighborhood $\O'$ of the bifurcation sphere $S_0$, a constant $\overline{c}=\overline{c}(\overline{A})>0$ such that $\O_{\overline{c}}\subseteq\O'$, and a smooth vector-field $\K$ in $\O'$ such that $\K=\underline{u}L-u\Lb$ on $(\HH^+\cup\HH^-)\cap\O'$,
\begin{equation}\label{ext1}
\mathcal{L}_{\K}\g=0,\quad[\T,\K]=0,\quad\K^\mu\sigma_\mu=0\quad\text{ in }\O',
\end{equation}
and
\begin{equation}\label{ext1.1}
\g(\K,\K)\leq-\overline{c}(r-1)^2\quad\text{ on }\Sigma_1\cap\O'.
\end{equation}
In addition, there is $\lambda_0\in\mathbb{R}$ such that the vector-field
\begin{equation*}
\Z=\T+\lambda_0\K
\end{equation*}
has complete periodic orbits in $\O'$.
\end{proposition}

The inequality \eqref{ext1.1} follows from \cite[Proposition
4.5]{AlIoKl} and \eqref{haw12.5}. In this section we extend $\K$
to the exterior region $\E$. The main result is the following:

\begin{theorem}\label{mainsecondstep}
The vector-field $\K$ constructed in $\O_{\overline{c}}$ can be extended to a smooth vector-field in the exterior region $\E$ such that
\begin{equation}\label{ext2}
\mathcal{L}_{\K}\g=0,\quad[\T,\K]=0,\quad\K^\mu\sigma_\mu=0\quad\text{ in }\E\cup\O_{\overline{c}}.
\end{equation}
\end{theorem}

The rest of this section is concerned with the proof of Theorem
\ref{mainsecondstep}. We construct the vector-field $\K$
recursively, in increasingly larger regions  defined  in terms of
the level sets of the function $y$. We rely  on   Carleman
estimates to prove, by an uniqueness argument similar to that of
\cite{AlIoKl}, that the extended $\K$ remains Killing at every
step in the process. The initial step is, of course,  that  given
by Proposition \ref{summary}.

Recall the definitions \eqref{ext4} and \eqref{ext5}, and the
identity $\mathcal{U}_R=\mathcal{V}_R$, see \eqref{alexnew2}.
Using the flow $\Psi_{t,\T}$ associated to $\T$ and the assumption
{\bf{GR}} on the orbits of $\T$, we define the connected open
space-time regions,
\begin{equation}\label{rot9}
\E_R=\{p\in\E:y(p)<R\}=\cup_{t\in\mathbb{R}}\Psi_{t,\T}(\mathcal{U}_R)\subseteq\E,\quad
R\geq y_0.
\end{equation}
Clearly, $ \E=\cup_{R\geq y_0}\E_R$. The main step in the proof
of the theorem is the following:

{\bf{Main Claim:}} For any $R\geq y_0$ there is a smooth
vector-field $\K$ defined in the connected open set $\E_R$, which
agrees with the vector-field $\K$ defined in Proposition
\ref{summary} in a neighborhood of $\mathcal{U}_{y_0}$ in $\E$,
such that
\begin{equation}\label{ext50}
\mathcal{L}_{\K}\g=0,\quad[\T,\K]=0,\quad\K^\mu\sigma_\mu=0\quad\text{ in }\E_R.
\end{equation}

The Main Claim follows for $R=y_0$ from Proposition \ref{summary}:
we define $\K$ in a small neighborhood of $\mathcal{U}_{y_0}$ in
$\E$ as in Proposition \ref{summary} and extend it to $\E_{y_0}$
by solving the ordinary differential equation $[\T,\K]=0$ (recall
that $\T$ does not vanish in $\E$). The remaining identities in
\eqref{ext50} hold on $\E_{y_0}$ since they hold in a small
neighborhood of $\mathcal{U}_{y_0}$ in $\E$ and $\T$ is
non-vanishing Killing vector-field.

Assume now that the Main Claim holds for some value $R_0\geq y_0$.
We would like to prove the Main Claim for some value
$R=R_0+\delta'$, for some
$\delta'=\delta'(\overline{A},\delta_0)>0$. We will use the
results and the notation in section \ref{propertiesy}.

Recall that $y,z,\sigma$ are smooth well-defined functions in $\E$
and $y+iz=(1-\sigma)^{-1}$. As in the proof of Lemma
\ref{Uincrease} let $Y^\al=\D^\al y$, which is a smooth
vector-field in $\E$. Using the last identity in \eqref{ext50},
$\K^\mu Y_\mu=0$ in $\E_{R_0}$. We compute in $\E_{R_0}$
\begin{equation*}
[\K,Y]_\be=\K^\al\D_\al Y_\be-Y^\al\D_\al\K_\be=\K^\al\D_\be\D_\al y+\D^\al y\D_\be\K_\al=\D_\be(\K^\al\D_\al y)=0.
\end{equation*}
Thus
\begin{equation}\label{ext51}
[\K,Y]=0\quad\text{ in }\E_{R_0}.
\end{equation}

For $R\geq y_0$ and $\delta>0$ small we define
\begin{equation*}
\widetilde{\O}_{\delta,R}=\cup_{p\in\delta_{\Sigma_1}(\mathcal{U}_R)}B_\delta(p).
\end{equation*}
Clearly, for $\delta$ sufficiently small and $R\geq y_0$
\begin{equation}\label{ext52}
\E_R\cap\widetilde{\O}_{\delta,R}=\{p\in\widetilde{\O}_{\delta,R}:y(p)<R\}.
\end{equation}
The vector-field $\K$ is defined in $\E_{R_0}\cap\widetilde{\O}_{\delta,R_0}$, by the induction hypothesis. We would like to extend it to the full open set $\widetilde{\O}_{\delta,R_0}$ as the solution of an ordinary differential equation of the form $[\K,\overline{Y}]=0$, where $\overline{Y}$ is a suitable vector-field in $\widetilde{\O}_{\delta,R_0}$. We summarize this construction in Lemma \ref{overY} below.

\begin{lemma}\label{overY}
There is a constant $\delta_3=\delta_3(\overline{A})>0$, a smooth
vector-field $\overline{Y}=\overline{Y}^\al\pr_\al$ in
$\widetilde{\O}_{\delta_3,R_0}$,
$\sum_{\al=0}^3|\overline{Y}^\al|\leq \delta_3^{-1}$ in
$\widetilde{\O}_{\delta_3,R_0}$, and a smooth extension of the
vector-field $\K$ (originally defined in
$\widetilde{\O}_{\delta_3,R_0}\cap\E_{R_0}$) to
$\widetilde{\O}_{\delta_3,R_0}$ such that
\begin{equation}\label{ext53}
\D_{\overline{Y}}\overline{Y}=0,\quad[\K,\overline{Y}]=0,\quad\overline{Y}(y)\geq\delta_3\qquad\text{ in }\widetilde{\O}_{\delta_3,R_0}.
\end{equation}
\end{lemma}

\begin{proof}[Proof of Lemma \ref{overY}] For $\delta$ sufficiently small we define
\begin{equation*}
S_{\delta,R_0}=\{x\in\widetilde{\O}_{\delta,R_0}:y(x)=R_0\}.
\end{equation*}
Clearly, $\delta_{\Sigma_1}(\mathcal{U}_{R_0})\subseteq S_{\delta,R_0}$. Since $y$ is a smooth function and $\D^\al y\D_\al y\geq\widetilde{C}^{-1}$ in $\widetilde{\O}_{\delta,R_0}$, the set $S_{\delta,R_0}$ is a smooth imbedded hypersurface. We define $\overline{Y}=Y$ on $S_{\delta,R_0}$, and extend $\overline{Y}$ to an open set of the form $\widetilde{\O}_{\delta',R_0}$, $\delta'\leq\delta$ by solving the geodesic equation $\D_{\overline{Y}}{\overline{Y}}=0$.

We first show that  $[\K,\overline{Y}]=0$ in
$\widetilde{\O}_{\delta'',R_0}\cap\E_{R_0}$,
$\delta''\in(0,\delta']$. Since $\K$ is tangent to
$S_{\delta,R_0}$, $\K(y)=0$,  and  $\overline{Y}=Y$ we deduce that
$[\K,\overline{Y}]=0$ along
  $S_{\delta,R_0}$.
On the other hand, we  have in $\widetilde{\O}_{\delta',R_0}\cap\E_{R_0}$ (where $\K$ is Killing),
\beaa
\D_{\overline{Y}}(\Lie_\K {\overline{Y}})=\Lie_\K(\D_{\overline{Y}}{\overline{Y}})-\D_{\Lie_\K {\overline{Y}}}\overline{Y} =-\D_{\Lie_\K {\overline{Y}}}\overline{Y}.
\eeaa
Thus,  $[\K,\overline{Y}]=0$ in $\widetilde{\O}_{\delta'',R_0}\cap\E_{R_0}$, $\delta''\in(0,\delta']$. We can now  extend $\K$ to $\widetilde{\O}_{\delta_3,R_0}$, $\delta_3\leq\delta''$, by solving the ordinary differential equation $[\K,\overline{Y}]=0$. This completes the proof of the lemma.
\end{proof}

We prove now that the vector field $\K$ is indeed a Killing vector-field (and verifies the other identities in \eqref{ext50}) in a small open set $\widetilde{\O}_{\delta,R_0}$. An argument of this type was used in \cite[Section 4]{AlIoKl}. For $|t|$ sufficiently small and $p_0\in \delta_{\Sigma_1}(\mathcal{U}_{R_0})$ we define, in a small neighborhood of $p_0$, the map $\Psi_{t,\K}$ obtained by flowing a parameter distance $t$ along the integral curves of $\K$. Let
\begin{equation*}
\g^t=\Psi_{t,\K}^\ast(\g)\qquad\text{ and }\qquad\T^t=\Psi^\ast_{t,\K}(\T).
\end{equation*}
The tensor $\g^t$ is a smooth Lorentz metric that satisfies the Einstein vacuum equations, and $\T^t$ is a smooth Killing vector-field for $\g^t$, in a small neighborhood of $p_0$ and for $|t|$ sufficiently small. In addition, since $\K$ is tangent to the hypersurface $\{y=R_0\}$, it follows from the induction hypothesis that $\g^t=\g$ and $\T^t=\T$ in a small neighborhood of $p_0$ intersected with $\E_{R_0}$. In addition, using the second identity in \eqref{ext53}, with $\Psi_t=\Psi_{t,\K}$,
\begin{equation*}
\frac{d}{dt}\Psi_{t}^\ast\overline{Y}=\lim_{h\to 0}\frac{\Psi^\ast_{t-h}\overline{Y}-\Psi^\ast_{t}\overline{Y}}{-h}=-\Psi_t^\ast\big(\lim_{h\to 0}\frac{\Psi_{-h}^\ast\overline{Y}-\Psi_0^\ast\overline{Y}}{-h}\big)=-\Psi_t^\ast(\mathcal{L}_\K\overline{Y})=0.
\end{equation*}
Thus $\Psi_{t}^\ast\overline{Y}=\overline{Y}$ and we infer that  ${\D^t}_{\overline{Y}}\overline{Y}=0$ in a small neighborhood of $p_0$, for $|t|$ sufficiently small, where $\D^t$ denotes the covariant derivative with respect to $\g^t$. The main step in proving the Main Claim is the following proposition:

\begin{proposition}\label{mainnew2}
Assume $p_0\in\delta_{\Sigma_1}(\mathcal{U}_{R_0})$, $\g'$ is a smooth Lorentz metric in $B_{\delta_4}(p_0)$, $\delta_4\in(0,\delta_3]$, such that $(B_{\delta_4}(p_0),\g')$ is a smooth Einstein vacuum spacetime, and $\T'$ is a smooth Killing vector-field for the metric $\g'$ in $B_{\delta_4}(p_0)$. In addition, assume that
\begin{equation*}
\begin{cases}
&\g'=\g\quad\text{ and }\quad\T'=\T\quad\text{ in }\E_{R_0}\cap B_{\delta_4}(p_0);\\
&\D'_{\overline{Y}}\overline{Y}=0\quad\text{ in }B_{\delta_4}(p_0),
\end{cases}
\end{equation*}
where $\D'$ denotes the covariant derivative induced by the metric $\g'$. Then $\g'=\g$ and $\T'=\T$ in $B_{\delta_5}(p_0)$ for some $\delta_5\in(0,\delta_4]$.
\end{proposition}

Assuming the proposition and Lemma \ref{extendedCarl2}, which we prove below, we complete now the proof of the Main Claim. It follows from Proposition  \ref{mainnew2} that $\K$ is a Killing vector-field in $B_{\delta_5}(p_0)$, for any $p_0\in\delta_{\Sigma_1}(\mathcal{U}_{R_0})$. In addition, since $\Psi^\ast_{t,\K}(\T)=\T$ for $|t|$ sufficiently small, it follows that $[\T,\K]=\mathcal{L}_{\K}\T=0$ in $B_{\delta_5}(p_0)$. Finally, in $B_{\delta_5}(p_0)$,
\begin{equation*}
\square_\g(\K^\mu\si_\mu)=\K^\mu\D^\al\D_\al\sigma_\mu=\K^\mu\D_\mu(\D^\al\si_\al)=-\mathcal{L}_{\K}(\FF^2)=0,
\end{equation*}
(using $\square_\g\K=0$, $\D\K$ is antisymmetric, $\D\sigma$ is symmetric, $\D^\al\sigma_\al=-\FF^2$, see \eqref{Ernst1}), and
\begin{equation*}
\T(\K^\mu\si_\mu)=\K(\T(\sigma))=0.
\end{equation*}
Since $\K^\mu\si_\mu=0$ in $B_{\delta_5}(p_0)\cap\E_{R_0}$ (the induction hypothesis), it follows from Lemma \ref{extendedCarl2} below, with $H=0$, that $\K^\mu\si_\mu=0$ in $B_{\delta_6}(p_0)$, $\delta_6\in(0,\delta_5]$.

To summarize, we proved that $\K$ extends to the open set
$\widetilde{\O}_{\delta_6,R_0}=\cup_{p_0\in\delta_{\Sigma_1}(\mathcal{U}_{R_0})}B_{\delta_6}(p_0)$,
$\delta_6=\delta_6(\overline{A},\delta_0)>0$, as a smooth vector,
and the identities in \eqref{ext50} hold in this set. Using the
inclusion \eqref{ext20}, it follows that $\K$ is well defined and
satisfies the identities \eqref{ext50} in a small neighborhood of
$\mathcal{U}_{R_0+\delta_6^2}$. Thus we can extend $\K$ to the
region $\E_{R_0+\delta_6^2}$, by solving the ordinary differential
equation $[\T,\K]=0$. The Main Claim follows.

\subsection{Proof of Proposition \ref{mainnew2}} We prove proposition \ref{mainnew2}
following the same scheme as in the proof of \cite[Proposition 4.3]{AlIoKl}.
We first  fix some smooth frames $v_{(1)}, v_{(2)}, v_{(3)}, v_{(4)}=\overline{Y}$ and ${v'}_{(1)}, {v'}_{(2)}, {v'}_{(3)}, {v'}_{(4)}=\overline{Y}$ in a small neighborhood $B_{\delta'}(p_0)$, such that, for $a=1,2,3,4$,
\begin{equation*}
\begin{split}
&\D_{\overline{Y}}v_{(a)}=0\text{ and }{\D'}_{\overline{Y}}{v'}_{(a)}=0\quad \text{ in }B_{\delta'}(p_0);\\
&v_{(a)}={v'}_{(a)}\quad\text{ in }\E_{R_0}\cap B_{\delta'}(p_0).
\end{split}
\end{equation*}
The idea of the proof is to derive ODE's for the differences $dv=v'-v$, $d\Gamma=\Gamma'-\Gamma$, $dT=\T'-\T$ and $dF=F'-F$, with source terms in $dR=\R'-\R$. We combine these ODE's with an equation for $\square_\g(dR)$ and equation for $\T(dR)$. Finally, we prove uniqueness of solutions of the resulting coupled system, see Lemma \ref{extendedCarl2}, using Carleman inequalities as in \cite{IoKl}, \cite{IoKl2}, \cite{Al}, \cite{AlIoKl}.

As in the proof of \cite[Proposition 4.3]{AlIoKl}, we define, for $a,b,c,d=1,\ldots 4$ and $\al,\be=0,\ldots,3$,
\begin{equation}\label{ext61}
\begin{split}
&(d\Gamma)_{(a)(b)(c)}=\Gamma'_{(a)(b)(c)}-{\Gamma}_{(a)(b)(c)}=\g'({v'}_{(a)},\D'_{{v'}_{(c)}}{v'}_{(b)})-\g(v_{(a)},\D_{v_{(c)}}v_{(b)});\\
&(\partial d\Gamma)_{\al(a)(b)(c)}=\pr_\al[(d\Gamma)_{(a)(b)(c)}];\\
&(dR)_{(a)(b)(c)(d)}=\R'({v'}_{(a)},{v'}_{(b)},{v'}_{(c)},{v'}_{(d)})-\R(v_{(a)},v_{(b)},v_{(c)},v_{(d)});\\
&(\partial dR)_{\al(a)(b)(c)(d)}=\pr_\al[(dR)_{(a)(b)(c)(d)}];\\
&(dv)_{(a)}^\be={v'}^\be_{(a)}-v^\be_{(a)}\qquad\text{ where }v_{(a)}=v^\be_{(a)}\pr_\be\text{ and }{v'}_{(a)}={v'}^\be_{(a)}\pr_\be;\\
&(\partial  dv)^\be_{\al(a)}=\pr_\al[(dv)_{(a)}^\be].
\end{split}
\end{equation}
As before, the coordinate  frame $\pr_0,\ldots,\pr_3$ is induced by the diffeomorphism $\Phi_1$. Let $\g_{(a)(b)}=\g(v_{(a)},v_{(b)})$, $\g'_{(a)(b)}=\g'({v'}_{(a)},{v'}_{(b)})$. The identities $\D_{\overline{Y}}v_{(a)}=\D'_{\overline{Y}}{v'}_{(a)}=0$ show that $\overline{Y}(\g_{(a)(b)})=\overline{Y}(\g'_{(a)(b)})=0$. Since $\g_{(a)(b)}=\g'_{(a)(b)}$ in $\E_{R_0}\cap B_{\delta'(p_0)}$ it follows  that
\begin{equation}\label{a9}
\g_{(a)(b)}=\g'_{(a)(b)}:=h_{(a)(b)}\,\text{ and }\,\overline{Y}(h_{(a)(b)})=0\text{ in }B_{\delta''}(p_0),
\end{equation}
for some constant $\delta''=\delta''(\overline{A},\delta_0)\in(0,\delta']$. Clearly, $\Gamma_{(a)(b)(4)}=\Gamma'_{(a)(b)(4)}=0$. We use now the definition of the Riemann curvature tensor to find a system of equations for $\overline{Y}[(d\Gamma)_{(a)(b)(c)}]$. We have
\begin{equation*}
\begin{split}
\R_{(a)(b)(c)(d)}&=\g(v_{(a)},\D_{v_{(c)}}(\D_{v_{(d)}}v_{(b)})-\D_{v_{(d)}}(\D_{v_{(c)}}v_{(b)})-\D_{[v_{(c)},v_{(d)}]}v_{(b)})\\
&=\g(v_{(a)},\D_{v_{(c)}}(\g^{(m)(n)}\Gamma_{(m)(b)(d)}v_{(n)}))-\g(v_{(a)},\D_{v_{(d)}}(\g^{(m)(n)}\Gamma_{(m)(b)(c)}v_{(n)}))\\
&+\g^{(m)(n)}\Gamma_{(a)(b)(n)}(\Gamma_{(m)(c)(d)}-\Gamma_{(m)(d)(c)})\\
&=v_{(c)}(\Gamma_{(a)(b)(d)})-v_{(d)}(\Gamma_{(a)(b)(c)})+\g^{(m)(n)}\Gamma_{(a)(b)(n)}(\Gamma_{(m)(c)(d)}-\Gamma_{(m)(d)(c)})\\
&+\g_{(a)(n)}[\Gamma_{(m)(b)(d)}v_{(c)}(\g^{(m)(n)})-\Gamma_{(m)(b)(c)}v_{(d)}(\g^{(m)(n)})]\\
&+\g^{(m)(n)}(\Gamma_{(m)(b)(d)}\Gamma_{(a)(n)(c)}-\Gamma_{(m)(b)(c)}\Gamma_{(a)(n)(d)}).
\end{split}
\end{equation*}
We set $d=4$ and use $\Gamma_{(a)(b)(4)}=v_{(4)}(\g^{(a)(b)})=0$ and $\g^{(a)(b)}=h^{(a)(b)}$; the result is
\begin{equation*}
\overline{Y}(\Gamma_{(a)(b)(c)})=-h^{(m)(n)}\Gamma_{(a)(b)(n)}\Gamma_{(m)(4)(c)}-\R_{(a)(b)(c)(4)}.
\end{equation*}
Similarly,
\begin{equation*}
\overline{Y}({\Gamma'}_{(a)(b)(c)})=-h^{(m)(n)}{\Gamma'}_{(a)(b)(n)}{\Gamma'}_{(m)(4)(c)}-{\R'}_{(a)(b)(c)(4)}.
\end{equation*}
We subtract these two identities to derive
\begin{equation}\label{a20}
\overline{Y}[(d\Gamma)_{(a)(b)(c)})]={}^{(1)}F_{(a)(b)(c)}^{(d)(e)(f)}(d\Gamma)_{(d)(e)(f)}-(dR)_{(a)(b)(c)(4)}
\end{equation}
for some smooth function ${}^{(1)}F$. This can be written schematically in the form
\begin{equation}\label{schem1}
\overline{Y}(d\Gamma)=\mathcal{M}_\infty(d\Gamma)+\mathcal{M}_\infty(d R).
\end{equation}
We will use such schematic equations for simplicity of notation\footnote{In general, given $H=(H_1,\ldots H_L):B_{\delta''}(p_0)\to\mathbb{R}^L$ we let $\mathcal{M}_\infty(H):B_{\delta''}(p_0)\to\mathbb{R}^{L'}$ denote vector-valued functions of the form ${\mathcal{M}_\infty(H)}_{l'}=\sum_{l=1}^LA_{l'}^lH_l$, where the coefficients $A_{l'}^l$ are smooth on $B_{\delta''}(p_0)$.}. By differentiating \eqref{schem1}, we also derive
\begin{equation}\label{schem2}
\overline{Y}(\partial  d\Gamma)=\mathcal{M}_\infty(d\Gamma)+\mathcal{M}_\infty(\partial d\Gamma)+\mathcal{M}_\infty(dR)+\mathcal{M}_\infty(\partial dR).
\end{equation}

With the notation in \eqref{ext61}, since $[v_{(4)},v_{(b)}]=-\D_{v_{(b)}}v_{(4)}=-{\Gamma^{(c)}}_{(4)(b)}v_{(c)}$, we have
\begin{equation*}
v_{(4)}^\al\pr_\al(v_{(b)}^\be)-v_{(b)}^\al\pr_\al(v_{(4)}^\be)=-{\Gamma}_{(a)(4)(b)}v_{(c)}^\be\g^{(a)(c)}.
\end{equation*}
Similarly,
\begin{equation*}
v_{(4)}^\al\pr_\al({v'}_{(b)}^\be)-{v'}_{(b)}^\al\pr_\al(v_{(4)}^\be)=-{\Gamma'}_{(a)(4)(b)}{v'}_{(c)}^\be{\g'}^{(a)(c)}.
\end{equation*}
We subtract these two identities to conclude that, schematically,
\begin{equation}\label{schem3}
\overline{Y}(dv)=\mathcal{M}_\infty(d\Gamma)+\mathcal{M}_\infty(dv).
\end{equation}
By differentiating \eqref{schem3} we also have
\begin{equation}\label{schem4}
\overline{Y}(\partial dv)=\mathcal{M}_\infty(d\Gamma)+\mathcal{M}_\infty(\partial d\Gamma)+\mathcal{M}_\infty(dv)+\mathcal{M}_\infty(\partial dv).
\end{equation}

We derive now a wave equation for $dR$. We start from the identity
\begin{equation*}
(\square_\g  \R)_{(a)(b)(c)(d)}-(\square_{\g'}{\R'})_{(a)(b)(c)(d)}=\mathcal{M}_\infty(dR),
\end{equation*}
which follows from the standard wave equations satisfied by $\R$ and $\R'$ and the fact that $\g^{(m)(n)}={\g'}^{(m)(n)}=h^{(m)(n)}$. We also have
\begin{equation*}
\begin{split}
&\D_{(m)}\R_{(a)(b)(c)(d)}-{\D'}_{(m)}{\R'}_{(a)(b)(c)(d)}\\
&=\mathcal{M}_\infty(dv)+\mathcal{M}_\infty(d\Gamma)+\mathcal{M}_\infty(dR)+\mathcal{M}_\infty(\partial dR).
\end{split}
\end{equation*}
It follows from the last two equations that
\begin{equation*}
\begin{split}
&\g^{(m)(n)}v_{(n)}(v_{(m)}(\R_{(a)(b)(c)(d)}))-{\g'}^{(m)(n)}{v'}_{(n)}({v'}_{(m)}({\R'}_{(a)(b)(c)(d)}))\\
&=\mathcal{M}_\infty(dv)+\mathcal{M}_\infty(d\Gamma)+\mathcal{M}_\infty(\partial d\Gamma)+\mathcal{M}_\infty(dR)+\mathcal{M}_\infty(\partial dR).
\end{split}
\end{equation*}
Since $\g^{(m)(n)}={\g'}^{(m)(n)}$ it follows that
\begin{equation*}
\begin{split}
&\g^{(m)(n)}v_{(n)}(v_{(m)}((dR)_{(a)(b)(c)(d)}))\\
&=\mathcal{M}_\infty(dv)+\mathcal{M}_\infty(\partial dv)+\mathcal{M}_\infty(d\Gamma)+\mathcal{M}_\infty(\partial d\Gamma)+\mathcal{M}_\infty(dR)+\mathcal{M}_\infty(\partial dR).
\end{split}
\end{equation*}
Thus
\begin{equation}\label{schem5}
\square_\g(dR)=\mathcal{M}_\infty(dv)+\mathcal{M}_\infty(\partial dv)+\mathcal{M}_\infty(d\Gamma)+\mathcal{M}_\infty(\partial d\Gamma)+\mathcal{M}_\infty(dR)+\mathcal{M}_\infty(\partial dR).
\end{equation}
This is our main wave equation.

We collect now equations \eqref{schem1}, \eqref{schem2}, \eqref{schem3}, \eqref{schem4}, and \eqref{schem5}:
\begin{equation}\label{ext60}
\begin{split}
&\overline{Y}(d\Gamma)=\mathcal{M}_\infty(d\Gamma)+\mathcal{M}_\infty(dR);\\
&\overline{Y}(\partial d\Gamma)=\mathcal{M}_\infty(d\Gamma)+\mathcal{M}_\infty(\partial d\Gamma)+\mathcal{M}_\infty(dR)+\mathcal{M}_\infty(\partial dR);\\
&\overline{Y}(dv)=\mathcal{M}_\infty(dv)+\mathcal{M}_\infty(d\Gamma);\\
&\overline{Y}(\partial dv)=\mathcal{M}_\infty(dv)+\mathcal{M}_\infty(\partial dv)+\mathcal{M}_\infty(d\Gamma)+\mathcal{M}_\infty(\partial d\Gamma);\\
&\square_\g(dR)=\mathcal{M}_\infty(dv)+\mathcal{M}_\infty(\partial dv)+\mathcal{M}_\infty(d\Gamma)+\mathcal{M}_\infty(\partial d\Gamma)+\mathcal{M}_\infty(dR)+\mathcal{M}_\infty(\partial dR).
\end{split}
\end{equation}
This is our first main system of equations.

We derive now an additional system of this type, to exploit the existence of the Killing vector-fields $\T$ and $\T'$. For $a,b=1,\ldots,4$ let
\begin{equation}\label{ext62}
\begin{split}
&(dT)_{(a)}={\T'}_{(a)}-\T_{(a)}=\g'(\T',L'_{(a)})-\g(\T,L_{(a)});\\
&(dF)_{(a)(b)}={F'}_{(a)(b)}-F_{(a)(b)}={\D'}_{(a)}{\T'}_{(b)}-\D_{(a)}\T_{(b)}.
\end{split}
\end{equation}
Using the identities $\D_{v_{(4)}}v_{(b)}=0$ and ${\D'}_{v_{(4)}}{v'}_{(b)}=0$ it follows that $v_{(4)}(\T_{(b)})=F_{(4)(b)}$ and $v_{(4)}({\T'}_{(b)})=F'_{(4)(b)}$. Thus
\begin{equation}\label{schem10}
\overline{Y}(dT)=\mathcal{M}_\infty(dF).
\end{equation}
We also have, using again $\D_{v_{(4)}}v_{(b)}=0$,
\begin{equation*}
v_{(4)}(F_{(a)(b)})=\D_{(4)}F_{(a)(b)}=\g^{(c)(d)}\T_{(d)}\R_{(c)(4)(a)(b)}=h^{(c)(d)}\T_{(d)}\R_{(c)(4)(a)(b)}.
\end{equation*}
Similarly,
\begin{equation*}
v_{(4)}({F'}_{(a)(b)})=h^{(c)(d)}{\T'}_{(d)}\R'_{(c)(4)(a)(b)}.
\end{equation*}
Thus, in our schematic notation,
\begin{equation}\label{schem11}
\overline{Y}(dF)=\mathcal{M}_\infty(\T)+\mathcal{M}_\infty(\R).
\end{equation}

Finally, we use the identities
\begin{equation*}
\begin{split}
0=(\mathcal{L}_{\T}\R)_{(a)(b)(c)(d)}&=\T^{(m)}\D_{(m)}\R_{(a)(b)(c)(d)}+\D_{(a)}\T^{(m)}\R_{(m)(b)(c)(d)}+\D_{(b)}\T^{(m)}\R_{(a)(m)(c)(d)}\\
&+\D_{(c)}\T^{(m)}\R_{(a)(b)(m)(d)}+\D_{(d)}\T^{(m)}\R_{(a)(b)(c)(m)},
\end{split}
\end{equation*}
and
\begin{equation*}
\begin{split}
0&=(\mathcal{L}_{\T'}\R')_{(a)(b)(c)(d)}={\T'}^{(m)}{\D'}_{(m)}{\R'}_{(a)(b)(c)(d)}+{\D'}_{(a)}{\T'}^{(m)}{\R'}_{(m)(b)(c)(d)}\\
&+{\D'}_{(b)}{\T'}^{(m)}{\R'}_{(a)(m)(c)(d)}+{\D'}_{(c)}{\T'}^{(m)}{\R'}_{(a)(b)(m)(d)}+{\D'}_{(d)}{\T'}^{(m)}{\R'}_{(a)(b)(c)(m)}.
\end{split}
\end{equation*}
Thus
\begin{equation*}
{\T'}^{(m)}{\D'}_{(m)}{\R'}_{(a)(b)(c)(d)}-\T^{(m)}\D_{(m)}\R_{(a)(b)(c)(d)}=\mathcal{M}_\infty(dF)+\mathcal{M}_\infty(dR),
\end{equation*}
which easily gives
\begin{equation}\label{schem12}
\T(dR)=\mathcal{M}_\infty(dF)+\mathcal{M}_\infty(dR)+\mathcal{M}_\infty(d\Gamma)+\mathcal{M}_\infty(dT)+\mathcal{M}_\infty(dv).
\end{equation}
We collect now equations \eqref{schem10}, \eqref{schem11}, and \eqref{schem12}, thus
\begin{equation}\label{schem51}
\begin{split}
&\overline{Y}(dT)=\mathcal{M}_\infty(dF);\\
&\overline{Y}(dF)=\mathcal{M}_\infty(dT)+\mathcal{M}_\infty(dR);\\
&\T(dR)=\mathcal{M}_\infty(dF)+\mathcal{M}_\infty(d\Gamma)+\mathcal{M}_\infty(dT)+\mathcal{M}_\infty(dL)+\mathcal{M}_\infty(dR).
\end{split}
\end{equation}
This is our second main system of differential equations. Since $\g'=\g$ and $\T'=\T$ in $\E_{R_0}\cap B_{\delta''}(p_0)$, the functions $d\Gamma,\partial d\Gamma, dv, \partial dv, dT, dF, dR$ vanish in $\E_{R_0}\cap B_{\delta''}(p_0)$. Therefore, using both systems \eqref{ext60} and \eqref{schem51}, the lemma is a consequence of Lemma \ref{extendedCarl2} below.

\begin{lemma}\label{extendedCarl2}
Assume $\delta>0$, $p_0\in\delta_{\Sigma_1}(\mathcal{U}_{R_0})$ and $G_i,H_j:B_{\delta}(p_0)\to\mathbb{R}$ are smooth functions, $i=1,\ldots,I$, $j=1,\ldots,J$. Let $G=(G_1,\ldots,G_I)$, $H=(H_1,\ldots,H_J)$, $\pr G=(\pr_0G_1,\ldots,\pr_4G_I)$ and assume that, in $B_{\delta}(p_0)$,
\begin{equation}\label{ext68}
\begin{cases}
&\square_\g G=\mathcal{M}_\infty(G)+\mathcal{M}_\infty(\pr G)+\mathcal{M}_\infty(H);\\
&\T(G)=\mathcal{M}_\infty(G)+\mathcal{M}_\infty(H);\\
&\overline{Y}(H)=\mathcal{M}_\infty(G)+\mathcal{M}_\infty(\pr G)+\mathcal{M}_\infty(H).
\end{cases}
\end{equation}
Assume that $G=0$ and $H=0$ in $B_{\delta}(p_0)\cap\E_{R_0}=\{x\in B_{\delta}(p_0):y(x)<R_0\}$. Then $G=0$ and $H=0$ in $B_{\widetilde{\delta}}(p_0)$ for some $\widetilde{\delta}\in(0,\delta)$ sufficiently small.
\end{lemma}

Unique continuation theorems of this type in the case $H=0$ were
proved by two of the authors in \cite{IoKl} and \cite{IoKl2},
using Carleman estimates. It is not hard to adapt the proofs,
using the same Carleman estimates, to the general case. The
essential ingredients are the $\T$-conditional pseudo-convexity
property in Lemma \ref{Tpseudo} and the inequality
$y(p_0)\geq(1+\sqrt{1-4B})/2+\widetilde{C}^{-1}$, see
\eqref{mainneed}. We provide all the details below.

\subsection{Proof of Lemma \ref{extendedCarl2}}
We will use a Carleman estimate proved by two of the authors in \cite[Section 3]{IoKl}, which we recall below. We may assume that the value of $\delta$ in Lemma \ref{extendedCarl2} is sufficiently small. For $r\leq\delta$ let $B_r=B_r(p_0)$ Notice that, if $\T=\T^\al\partial_\al,\overline{Y}=\overline{Y}^\al\partial_\al$  in the coordinate frame induced by the diffeomorphism $\Phi_1$ then
\begin{equation}\label{po1}
\sup_{x\in B_\delta}\sum_{j=0}^4\sum_{\al=0}^3(|\partial^j\T^\al(x)|+|\partial^j\overline{Y}^\al(x)|)
\leq \widetilde{C}=\widetilde{C}(\overline{A}).
\end{equation}

\begin{definition}\label{psconvex}
A family of weights $h_\eps:B_{\eps^{10}}\to\mathbb{R}_+$, $\eps\in(0,\eps_1)$, $\eps_1\leq\delta$ will be called $\T$-conditional pseudo-convex if for any $\eps\in(0,\eps_1)$
\begin{equation}\label{po5}
\begin{split}
h_\eps(p_0)=\eps,\quad\sup_{x\in B_{\eps^{10}}}\sum_{j=1}^4\eps^j|\partial^jh_\eps(x)|\leq\eps/\eps_1,\quad |\T(h_\eps)(p_0)|\leq\eps^{10},
\end{split}
\end{equation}
\begin{equation}\label{po3.2}
\D^\alpha h_\eps(p_0)\D^\be h_\eps(p_0)(\D_\al h_\eps\D_\be h_\eps-\eps\D_\al\D_\be h_\eps)(p_0)\geq\eps_1^2,
\end{equation}
and there is $\mu\in[-\eps_1^{-1},\eps_1^{-1}]$ such that for all vectors $X=X^\alpha\partial_\alpha$ at $p_0$
\begin{equation}\label{po3}
\begin{split}
&\eps_1^2[(X^1)^2+(X^2)^2+(X^3)^2+(X^4)^2]\\
&\leq X^\al X^\be(\mu\g_{\al\be}-\D_\al\D_\be h_\eps)(p_0)+\eps^{-2}(|X^\al \T_\al(p_0)|^2+|X^\al\D_\al h_\eps(x_0)|^2).
\end{split}
\end{equation}
A function $e_\eps:B_{\eps^{10}}\to\mathbb{R}$ will be called a negligible perturbation if
\begin{equation}\label{smallweight}
\sup_{x\in B_{\eps^{10}}}|\partial^je_\eps(x)|\leq\eps^{10}\qquad\text{ for  }j=0,\ldots,4.
\end{equation}
\end{definition}

Our main Carleman estimate, see \cite[Section 3]{IoKl}, is the
following:

\begin{lemma}\label{Cargen}
Assume $\eps_1\leq \delta$, $\{h_\eps\}_{\eps\in(0,\eps_1)}$ is a $\T$-conditional pseudo-convex family, and $e_\eps$ is a negligible perturbation for any $\eps\in(0,\eps_1]$. Then there is $\eps\in (0,\eps_1)$ sufficiently small (depending only on $\eps_1$) and $C$ sufficiently large such that for any $\lambda\geq C$ and any $\phi\in C^\infty_0(B_{\eps^{10}})$
\begin{equation}\label{Car1gen}
\lambda \|e^{-\lambda f_\eps}\phi\|_{L^2}+\|e^{-\lambda f_\eps}|\partial \phi|\,\|_{L^2}\leq C\lambda^{-1/2}\|e^{-\lambda f_\eps}\,\square_{\g}\phi\|_{L^2}+\eps^{-6}\|e^{-\lambda f_\eps}\T(\phi)\|_{L^2},
\end{equation}
where $f_\ep=\ln (h_\eps+e_\eps)$.
\end{lemma}

We also need a Carleman inequality to exploit the last equation in
\eqref{ext68}.

\begin{lemma}\label{CarODE}
Assume $\eps\leq \delta$ is sufficiently small, $e_\eps$ is a negligible perturbation, and $h_\eps:B_{\eps^{10}}\to\mathbb{R}_+$ satisfies
\begin{equation}\label{ODEneeds}
h_{\eps}(p_0)=\eps,\quad\sup_{x\in B_{\eps^{10}}}\sum_{j=1}^2\eps^j|\partial^jh_{\eps}(x)|\leq 1,\quad |\overline{Y}(h_\eps)(p_0)|\geq \eps.
\end{equation}
Then there is $C$ sufficiently large such that for any $\lambda\geq C$ and any $\phi\in C^\infty_0(B_{\eps^{10}})$
\begin{equation}\label{Car1ODE}
\|e^{-\lambda f_\eps}\phi\|_{L^2}\leq 4(\eps\lambda)^{-1}\|e^{-\lambda f_\eps}\overline{Y}(\phi)\|_{L^2},
\end{equation}
where $f_\ep=\ln (h_\eps+e_\eps)$.
\end{lemma}

This inequality was proved in \cite[Appendix A]{AlIoKl}. See also \cite[Chapter 28]{Ho} for much more general Carleman inequalities under suitable pseudo-convexity conditions.

To prove Lemma \ref{extendedCarl2} we set
\begin{equation}\label{mw1}
h_\eps=y-y(p_0)+\eps\quad\text{and }\quad e_\eps=\eps^{12}N^{p_0},
\end{equation}
where $N^{p_0}(x)=|\Phi_1^{-1}(x)- \Phi_1^{-1}(p_0)|^2$ is the square of the standard euclidean norm.

It is clear that $e_\eps$ is a negligible perturbation, in the
sense of \eqref{smallweight}, for $\eps$ sufficiently small. Also,
it is clear that $h_\eps$ verifies the condition \eqref{ODEneeds},
for $\eps$ sufficiently small, see Lemma \ref{overY}.

We show now that there is $\eps_1=\eps_1(\delta)$ sufficiently small such that the family of weights $\{h_\eps\}_{\eps\in(0,\eps_1)}$ is $\T$-conditional pseudo-convex, in the sense of Definition \ref{psconvex}. Condition \eqref{po5} is clearly satisfied, since $\T(y)=0$. Condition \ref{po3.2} is also satisfied for $\eps$ sufficiently small since $\D^\al y(p_0)\D_\al y(p_0)\geq \widetilde{C}^{-1}$, see \eqref{concly}. To prove \eqref{po3} for some vector $X$ we apply Lemma \ref{Tpseudo} if $|X^\al\T_\al|+|X^\al\D_\al y|\leq c_2|X|$; if $|X^\al\T_\al|+|X^\al\D_\al y|\geq c_2|X|$ then the second term in the right-hand side of \eqref{po3} dominates the other terms, provided that $\eps_1$ is sufficiently small.

It follows from the Carleman estimates in Lemmas \ref{Cargen} and \ref{CarODE} that there is $\eps=\eps(\delta,\overline{A})>0$ and a constant $C=C(\delta,\overline{A})\geq 1$ such that
\begin{equation}\label{Bar1}
\begin{split}
&\lambda \|e^{-\lambda f_\eps}\phi\|_{L^2}+\|e^{-\lambda f_\eps}|\partial\phi|\,\|_{L^2}\leq C\lambda^{-1/2}\|e^{-\lambda f_\eps}\,\square_{\g}\phi\|_{L^2}+C\|e^{-\lambda f_\eps}\T(\phi)\|_{L^2};\\
&\lambda^{1/2}\|e^{-\lambda f_\eps}\phi\|_{L^2}\leq C\lambda^{-1/2}\|e^{-\lambda f_\eps}\overline{Y}(\phi)\|_{L^2},
\end{split}
\end{equation}
for any $\phi\in C^\infty_0(B_{\eps^{10}}(p_0))$ and any $\lambda\geq C$, where $f_\ep=\ln (h_\eps+e_\eps)$. Let $\eta:\mathbb{R}\to[0,1]$ denote a smooth function supported in $[1/2,\infty)$ and equal to $1$ in $[3/4,\infty)$. For $i=1,\ldots,I$, $j=1,\ldots J$ we define,
\begin{equation}\label{pr2}
\begin{split}
&G^{\eps}_i=G_i\cdot \big(1-\eta(N^{x_0}/\eps^{20})\big)=G_i\cdot \widetilde{\eta}_{\eps}\\
&H^{\eps}_j=H_j\cdot \big(1-\eta(N^{x_0}/\eps^{20})\big)=H_j\cdot \widetilde{\eta}_{\eps}.
\end{split}
\end{equation}
Clearly, $G^{\eps}_i,H^{\eps}_j\in C^\infty _0(B_{\eps^{10}}(p_0))$. We would like to apply the inequalities in \eqref{Bar1} to the functions $G^{\eps}_i,H^{\eps}_j$, and then let $\lambda \to\infty$.

Using the definition \eqref{pr2}, we have
\begin{equation*}
\begin{split}
&\square_\g G^{\eps}_i=\widetilde{\eta}_{\eps}\square_\g G_i+2\D_\al G_i\D^\al \widetilde{\eta}_{\eps}+G_i\square_\g\widetilde{\eta}_{\eps};\\
&\T(G_i^\eps)=\widetilde{\eta}_\eps\T(G_i)+\T(\widetilde{\eta}_\eps)G_i;\\
&\overline{Y}(H^{\eps}_j)=\widetilde{\eta}_{\eps}\cdot\overline{Y}(H_j)+H_j\cdot\overline{Y}(\widetilde{\eta}_{\eps}).
\end{split}
\end{equation*}
Using the Carleman inequalities \eqref{Bar1}, for any $i=1,\ldots,I$, $j=1,\ldots,J$ we have
\begin{equation}\label{va10}
\begin{split}
&\lambda\cdot \|e^{-\lambda f_{\eps}}\cdot\widetilde{\eta}_{\eps}G_i\|_{L^2}+\|e^{-\lambda f_{\eps}}\cdot \widetilde{\eta}_{\eps}|\partial G_i|\, \|_{L^2}\\
&\leq C\lambda ^{-1/2}\cdot \|e^{-\lambda f_{\eps}}\cdot \widetilde{\eta}_{\eps}\square_\g G_i\|_{L^2}+C\|e^{-\lambda f_{\eps}}\cdot \widetilde{\eta}_{\eps}\T(G_i)\|_{L^2}\\
&+C'\Big[\|e^{-\lambda f_{\eps}}\cdot \D_\al G_i\D^\al \widetilde{\eta}_{\eps} \|_{L^2}+\|e^{-\lambda f_{\eps}}\cdot G_i( |\square_\g\widetilde{\eta}_{\eps}|+|\partial\widetilde{\eta}_{\eps}| )\|_{L^2}\Big]
\end{split}
\end{equation}
and
\begin{equation}\label{va10.1}
\lambda^{1/2}\|e^{-\lambda f_{\eps}}\cdot \widetilde{\eta}_{\eps}H_j\|_{L^2}\leq C\lambda^{-1/2}\|e^{-\lambda f_{\eps}}\cdot \widetilde{\eta}_{\eps}\overline{Y}(H_j)\|_{L^2}+C'\lambda^{-1/2}\|e^{-\lambda f_{\eps}}\cdot H_j |\partial\widetilde{\eta}_{\eps}|\|_{L^2},
\end{equation}
for any $\lambda\geq C$ and some constant $C'=C'(\overline{A},C)$. Using the main identities \eqref{ext68}, in $B_{\eps^{10}}(p_0)$ we estimate pointwise
\begin{equation}\label{va11}
\begin{split}
&|\square_\g G_i|\leq M\sum_{l=1}^I\big(|\partial G_l|+|G_l|\big)+M\sum_{m=1}^J|H_m|,\\
&|\T(G_i)|\leq M\sum_{l=1}^I|G_l|+M\sum_{m=1}^J|H_m|\\
&|\overline{Y}(H_j)|\leq M\sum_{l=1}^I\big(|\partial G_l|+|G_l|\big)+M\sum_{m=1}^J|H_m|,
\end{split}
\end{equation}
for some large constant $M$. We add inequalities \eqref{va10} and \eqref{va10.1} over $i,j$. The key observation is that, in view of \eqref{va11}, the main terms in the right-hand sides of \eqref{va10} and \eqref{va10.1} can be absorbed into the left-hand sides for $\lambda$ sufficiently large. Thus, for any $\lambda$ sufficiently large,
\begin{equation*}
\begin{split}
&\lambda\sum_{i=1}^I\|e^{-\lambda f_{\eps}}\widetilde{\eta}_{\eps}G_i\|_{L^2}+\sum_{i=1}^I\|e^{-\lambda f_{\eps}}\widetilde{\eta}_{\eps}|\partial G_i|\|_{L^2}+\lambda^{1/2}\sum_{j=1}^J\|e^{-\lambda f_{\eps}}\widetilde{\eta}_{\eps}H_j\|_{L^2}\\
&\leq C''\sum_{j=1}^J\|e^{-\lambda f_{\eps}}H_j|\partial \widetilde{\eta}_{\eps}|\|_{L^2}+C''\sum_{i=1}^I\big[\|e^{-\lambda f_{\eps}}\D_\al G_i\D^\al \widetilde{\eta}_{\eps} \|_{L^2}+\|e^{-\lambda f_{\eps}}G_i( |\square_\g\widetilde{\eta}_{\eps}|+|\partial\widetilde{\eta}_{\eps}| )\|_{L^2}\big].
\end{split}
\end{equation*}
We obseve that the functions $\square_\g\widetilde{\eta}_{\eps}$ and $\partial\widetilde{\eta}_{\eps}$ are supported in the set $\{x\in B_{\eps^{10}}(p_0):N^{p_0}\geq \eps^{20}/2\}$ and $\widetilde{\eta}_{\eps}=1$ in $B_{\eps^{100}}(p_0)$. By assumption, the functions $G_i,|\partial G_i|, H_j$ are supported in $\{x\in B_\delta(p_0):y(x)\geq y(p_0)\}$. In addition,
\begin{equation*}
\inf_{B_{\eps^{100}}(p_0)}\,e^{-\lambda f_{\eps}}\geq e^{\lambda/C'''}\sup_{\{x\in B_{\eps^{10}}(p_0):N^{p_0}\geq \eps^{20}/2\text{ and }y(x)\geq y(p_0)\}}\,e^{-\lambda f_{\eps}},
\end{equation*}
which follows easily from the definition \eqref{mw1} We let now $\lambda\to\infty$, as in \cite[Section 8]{IoKl}, to conclude that $\mathbf{1}_{B_{\eps^{100}}}\, G_i=0$ and $\mathbf{1}_{B_{\eps^{100}}}\, H_j=0$. The lemma follows.

\section{Construction of the rotational Killing vector-field $\Z$}\label{rotkilling}

In this section we extend the rotational Killing vector-field $\Z$ constructed in a small neighborhood of $S_0$, see Proposition \ref{summary}, to the entire exterior region $\E$. In $\E\cup \O_{\overline{c}}$ we define
\begin{equation*}
\Z=\T+\lambda_0\K,
\end{equation*}
where $\lambda_0$ is as in Proposition \ref{summary}. Clearly
$\lambda_0\neq 0$, in view of the assumption {\bf{GR}} that $\T$
does not vanish in $\E$, and $\Z$ does not vanish identically in
$\E$, since, by assumption {\bf{SBS}}, $\T$ does not vanish
identically on $S_0$. It follows from \eqref{ext2} that
\begin{equation}\label{rot1}
\mathcal{L}_\Z\g=0,\quad [\T,\Z]=[\K,\Z]=0,\quad \Z^\mu\sigma_\mu=0\quad\text{ in }\E\cup\O_{\overline{c}}.
\end{equation}
As in the proof of \eqref{ext51}, it follows that
\begin{equation}\label{rot2}
[\Z,Y]=0\quad\text{ in }\E\cup\O_{\overline{c}}.
\end{equation}

In view of Proposition \ref{summary}, there is
$t_0>0$\footnote{Using the assumption that the orbits of $\T$ in
$\E$ are complete and intersect $\Sigma^0$, see assumption
{\bf{GR}}, it is easy to see that any smooth vector-field $V$ in
$\E\cup\O_{\overline{c}}$ which commutes with $\T$ and is tangent
to $\HH^{\pm}\cap\O_{\overline{c}}$ has complete orbits in $\E$.}
such that $\Psi_{t_0,\Z}=\mathrm{Id}$ in $\O'$. Clearly
\begin{equation}\label{rot5}
\Psi_{s,\T}(p)=\Psi_{s,\T}\Psi_{t_0,\Z}(p)=
\Psi_{t_0,\Z}\Psi_{s,\T}(p)\quad\text{  for any }\quad
p\in\O_{\overline{c}}\cap\E\quad\text{ and  }\quad s\in\mathbb{R},
\end{equation}
using the commutation relation $[\T,\Z]=0$. It follows that
$\Psi_{t_0,\Z}(p)=p$ for any  $p\in \E_{y_0}$, recall definition
\eqref{rot9}. To prove this identity for any point $p\in\E$,
assume that
\begin{equation*}
\Psi_{t_0,\Z}(p)=p\quad\text{ for any }p\in\E_{R_0},
\end{equation*}
for some $R_0\geq y_0$. As before, it follows that
$\Psi_{t_0,\Z}(p)=p$ for any $p\in\E_{R_0}$. Using $[Y,\Z]=0$ and
an identity similar to \eqref{rot5}, it follows that
\begin{equation*}
\Psi_{t_0,\Z}(p)=p\text{ for any }p\in\E_{R_0+\delta'},
\end{equation*}
for some $\delta'=\delta'(\overline{A})>0$. To summarize, we
proved:

\begin{corollary}
There is a nontrivial smooth vector-field $\Z$ in $\E\cup\O_{\overline{c}}$, tangent to $\HH^+\cap\O_{\overline{c}}$ and $\HH^-\cap\O_{\overline{c}}$,  and a real number $t_0>0$ such that
\begin{equation*}
\Psi_{t_0,\Z}=\mathrm{Id},\quad\mathcal{L}_\Z\g=0,\quad [\T,\Z]=0,\quad \Z^\mu\sigma_\mu=0\quad\text{ in }\E.
\end{equation*}
\end{corollary}

\subsection{The time-like span of the two Killing fields.} We define the area function
\begin{equation*}
W=-\g(\T,\T)\g(\Z,\Z)+\g(\T,\Z)^2.
\end{equation*}
In this subsection we show that $W\geq 0$ in $\E$. More precisely, we prove the following
slightly stronger proposition:

\begin{proposition}\label{lintime}
The vector-field $\K$ constructed in Theorem \ref{mainsecondstep} does not vanish
at any point in $\E$. In addition, at any point $p\in\E$ there is a timelike linear combination of the vector-fields $\T$ and $\K$.
\end{proposition}

\begin{proof}[Proof of Proposition \ref{lintime}] In view of \eqref{ext1.1}, $\K$ does not vanish at any point in $\mathcal{U}_{y_0}$. It follows that $\K$ does not vanish at any point in $\E_{y_0}$, since $\K$ is constructed as the solution of $[\T,\K]=0$ in $\E_{y_0}$.

To prove that $\K$ does not vanish at any point $p\in\E$ we use
the identity $[\K,Y]=0$ in $\E$, see \eqref{ext51}. Let
\begin{equation*}
R_1=\sup\{R\in[y_0,\infty):\K\text{ does not vanish at any point in }\E_R\}.
\end{equation*}
If $R_1<\infty$ then $\K$ has to vanish at some point
$p_0\in\delta_{\Sigma_1}(\mathcal{U}_{R_1})$ (using the assumption
that any orbit of $\T$ in $\E$ intersects $\Sigma_1$, and the
observation that the set of points in $\E$ where $\K$ vanishes can
only be a union of orbits of $\T$). Since $[\K,Y]=0$ in $\E$, $\K$
vanishes on the integral curve $\gamma_{p_0}(t)$, $|t|\ll 1$, of
the vector-field $Y$ starting at the point $p_0$. However, this
integral curve intersects the set $\E_{R'}$ for some $R'<R_1$, in
contradiction with the definition of $R_1$. Thus $\K$ does not
vanish at any point in $\E$.

We prove now the second part of the proposition. Let
\begin{equation*}
N=\{p\in\E:\text{ there is no timelike linear combination of }\T\text{ and }\K\text{ at }p\}.
\end{equation*}
Clearly, the set $N$ is closed in $\E$ and consists of orbits of
the vector-field $\T$. In addition,
$N\subseteq\E\setminus\E_{y_0}$, since $\K$ itself is timelike in
$\E_{y_0}$ (see \eqref{ext1.1}).  In view of \eqref{mainneed} and
\eqref{ext5.5} we have $y\geq
(1+\sqrt{1-4B})/2+\widetilde{C}_2^{-2}$ in $\E\setminus\E_{y_0}$.
On the other hand $\g(\T,\T)=y/(y^2+z^2)-1$, hence $z^2\le
y-y^2=-(y^2-y+B)+B$ in $N$. Consequently, for some constant
$\widetilde{C}=\widetilde{C}(\overline{A})\gg 1$,
\begin{equation}\label{ext75}
y\geq (1+\sqrt{1-4B})/2+\widetilde{C}^{-1}\quad\text{ and }\quad B-z^2\geq \widetilde{C}^{-1}\quad\text{ in }N,
\end{equation}
Consider now  the set of  vector-fields $\T, \K$ as well as the gradient vector-fields
$ Y=\D^\al y\D_\al $,  $ Z=\D^\al z\D_\al$
at some point $p\in N$. Since $\T(\sigma)=\K(\sigma)=0$ we have
\begin{equation*}
\g(\T,Y)=\g(\T,Z)=\g(\K,Y)=\g(\K,Z)=0.
\end{equation*}
In addition, using \eqref{concly}, \eqref{se10}, and \eqref{ext75}
\begin{equation}\label{ext74}
\g(Y,Y)\geq \widetilde{C}^{-1},\qquad\g(Z,Z)\geq \widetilde{C}^{-1},\qquad|\g(Y,Z)|\leq\widetilde{C}\overline{\varep}^{1/5}\qquad\text{ in }N,
\end{equation}
for some constant $\widetilde{C}=\widetilde{C}(\overline{A})$. Since the metric $\g$ is Lorentzian, it follows that the vectors $\T,\K,Y,Z$ cannot be linearly independent at any point $p\in N$ (if they were linearly independent then the determinant of the matrix formed by the coefficients $\g(\T,\T)$, $\g(\T,\K)$, $\g(\K,\K)$ would have to be negative, in contradiction with $p\in N$).  Since  the triplets  $\T, Y, Z$ and $\K, Y, Z$ are linearly independent it must follows that $\K, \T$ are linearly dependent at points of $N$.  Thus
\begin{equation}\label{lo2}
\text{ for any }p\in N\text{ there is }a\in\mathbb{R}\text{ such that }\K_p=a\T_p\text{ and }\g(\T,\T)|_p\geq 0.
\end{equation}

We prove now that $N=\emptyset$. Assume that $N\neq\emptyset$ and let $p_0$ denote a point in $N$ such that $y(p_0)=\inf_{p\in N}y(p)$. Such a point exists since $N\cap \Sigma_1\subseteq\Sigma_{1+(4\widetilde{C}_1\widetilde{C}_2)^{-1/2}}$ (see \eqref{ext1.1}) is compact (observe that $\T$ is timelike in $\Si_R$ for large $R$). We may assume that $p_0\in N\cap \Sigma_1$. In view of \eqref{lo2}, there is $a_0\in\mathbb{R}$ such that $\K_{p_0}-a_0\T_{p_0}=0$. We look at the integral curve $\{\gamma_{p_0}(t):|t|\ll 1\}$ of the vector field $Y$ passing through $p_0$. Since $[Y,\K-a_0\T]=0$ and $Y_{p_0}\neq 0$, it follows that $\K=a_0\T$ in the set $\gamma_{p_0}(t)$, $|t|\ll 1$. Since $Y(y)=\g(Y,Y)$ is strictly positive at $p_0$ (see \eqref{ext74}), it follows that $y(\gamma_{p_0}(t))<y(p_0)$ if $t\in(-\widetilde{C}^{-1},0)$. Since $y(p_0)=\inf_{p\in N}y(p)$ (the definition of $p_0$), it follows that
\begin{equation*}
N\cap \{\gamma_{p_0}(t):t\in(-\widetilde{C}^{-1},0)\}=\emptyset.
\end{equation*}
Since $\K=a_0\T$ in $\{\gamma_{p_0}(t):t\in(-\widetilde{C}^{-1},0)\}$ it follows that $\g(\T,\T)<0$ in $\{\gamma_{p_0}(t):t\in(-\widetilde{C}^{-1},0)\}$. Using the formula $\g(\T,\T)=y/(y^2+z^2)-1$, it follows that the function $y-y^2-z^2$ vanishes at $p_0$ and is strictly negative on $\{\gamma_{p_0}(t):t\in(-\widetilde{C}^{-1},0)\}$. Thus
\begin{equation*}
Y(y-y^2-z^2)\geq 0\quad\text{ at }p_0.
\end{equation*}
On the other hand, using \eqref{ext74} and \eqref{ext75},
\begin{equation*}
\D^\al y\D_\al(y-y^2-z^2)=(1-2y)\D^\al y\D_\al y-2z\D^\al y\D_\al z<0\quad\text{ at }p_0,
\end{equation*}
provided that $\overline{\varep}$ is sufficiently small. This provides a contradiction.
\end{proof}

\appendix

\section{Asymptotic identities}\label{asymptot}

Recall, see assumption {\bf{GR}}, that we assumed the existence of an open subset $\M^{(end)}$ of $\M$ which is diffeomorphic to $\mathbb{R}\times(\{x\in\mathbb{R}^3:|x|>R\})$ for some $R$ sufficiently large. In local coordinates $\{t,x^i\}$ defined by this diffeomorphism, we assume that $\T=\partial_t$ and, with $r=\sqrt{(x^1)^2+(x^2)^2+(x^3)^2}$,
\begin{equation}\label{As-Flat5}
\g_{00}=-1+\frac{2M}{r}+O(r^{-2}),\quad \g_{ij}=\delta_{ij}+O(r^{-1}),\quad\g_{0i}=-\ep_{ijk}\frac{2S^jx^k}{r^3}+O(r^{-3}),
\end{equation}
for some $M>0$, $S^1,S^2,S^3\in\mathbb{R}$. Clearly,
\begin{equation}\label{e1}
\g^{00}=-1+O(r^{-1}),\quad \g^{ij}=\delta_{ij}+O(r^{-1}),\quad \g^{0i}=O(r^{-2}).
\end{equation}
We compute
\begin{equation}\label{e2}
F_{\al\be}=\D_\al\T_\be=\partial_\al(\g_{0\be})-\g(\partial_0,\D_{\partial_\al}\partial_\be)=\frac{1}{2}(\partial_\al\g_{0\be}-\partial_\be\g_{0\al}).
\end{equation}
Thus, using \eqref{As-Flat5}, for $j=1,2,3$,
\begin{equation}\label{e4}
F_{0j}=-(1/2)\partial_j\g_{00}=Mx^jr^{-3}+O(r^{-3}).
\end{equation}
We have
\begin{equation}\label{e5}
\begin{split}
&\g_{01}=2r^{-3}(S^3x^2-S^2x^3)+O(r^{-3})\\
&\g_{02}=2r^{-3}(S^1x^3-S^3x^1)+O(r^{-3})\\
&\g_{03}=2r^{-3}(S^2x^1-S^1x^2)+O(r^{-3}).
\end{split}
\end{equation}
Thus
\begin{equation}\label{e6}
\begin{split}
&F_{12}=(1/2)(\partial_1\g_{02}-\partial_2\g_{01})=S^3r^{-3}-3r^{-5}x^3(S^1x^1+S^2x^2+S^3x^3)+O(r^{-4}),\\
&F_{23}=(1/2)(\partial_2\g_{03}-\partial_3\g_{02})=S^1r^{-3}-3r^{-5}x^1(S^1x^1+S^2x^2+S^3x^3)+O(r^{-4}),\\
&F_{31}=(1/2)(\partial_3\g_{01}-\partial_1\g_{03})=S^2r^{-3}-3r^{-5}x^2(S^1x^1+S^2x^2+S^3x^3)+O(r^{-4}).
\end{split}
\end{equation}

We have
\begin{equation*}
{\dual F}_{\al\be}=(1/2)\in_{\al\be\mu\nu}F_{\rho\si}\g^{\mu\rho}\g^{\nu\si}.
\end{equation*}
Thus, using \eqref{e1}, \eqref{e4}, \eqref{e6},
\begin{equation}\label{e10}
\begin{split}
&{\dual F}_{01}=F_{23}+O(r^{-4})=S^1r^{-3}-3r^{-5}x^1(S^1x^1+S^2x^2+S^3x^3)+O(r^{-4}),\\
&{\dual F}_{02}=F_{31}+O(r^{-4})=S^2r^{-3}-3r^{-5}x^2(S^1x^1+S^2x^2+S^3x^3)+O(r^{-4}),\\
&{\dual F}_{03}=F_{12}+O(r^{-4})=S^3r^{-3}-3r^{-5}x^3(S^1x^1+S^2x^2+S^3x^3)+O(r^{-4}),\\
&{\dual F}_{12}=-F_{03}+O(r^{-3})=-Mx^3r^{-3}+O(r^{-3}),\\
&{\dual F}_{23}=-F_{01}+O(r^{-3})=-Mx^1r^{-3}+O(r^{-3}),\\
&{\dual F}_{31}=-F_{02}+O(r^{-3})=-Mx^2r^{-3}+O(r^{-3}).
\end{split}
\end{equation}
As a consequence,
\begin{equation}\label{e7}
\FF^2=(F_{\al\be}+i{\dual F}_{\al\be})(F^{\al\be}+i{\dual F}^{\al\be})=-4M^2r^{-4}+O(r^{-5}).
\end{equation}

By definition,
\begin{equation*}
\si_\mu=2\T^\al(F_{\al\mu}+i{\dual F}_{\al\mu}).
\end{equation*}
Thus
\begin{equation}\label{e15}
\begin{split}
&\si_0=0;\\
&\si_1=2Mx^1r^{-3}+O(r^{-3})+2i[S^1r^{-3}-3r^{-5}x^1(S^1x^1+S^2x^2+S^3x^3)+O(r^{-4})];\\
&\si_2=2Mx^2r^{-3}+O(r^{-3})+2i[S^2r^{-3}-3r^{-5}x^2(S^1x^1+S^2x^2+S^3x^3)+O(r^{-4})];\\
&\si_3=2Mx^3r^{-3}+O(r^{-3})+2i[S^3r^{-3}-3r^{-5}x^3(S^1x^1+S^2x^2+S^3x^3)+O(r^{-4})].
\end{split}
\end{equation}
Thus
\begin{equation}\label{e16}
\sigma=1-2Mr^{-1}+O(r^{-2})+i[2r^{-3}(S^1x^1+S^2x^2+S^3x^3)+O(r^{-3})].
\end{equation}
Thus
\begin{equation*}
y+iz=(1-\sigma)^{-1}=\frac{r}{2M}+O(1)+i\Big[\frac{S^1x^1+S^2x^2+S^3x^3}{2M^2r}+O(r^{-1})\Big],
\end{equation*}
which gives
\begin{equation}\label{e17}
y=\frac{r}{2M}+O(1),\qquad z=\frac{S^1x^1+S^2x^2+S^3x^3}{2M^2r}+O(r^{-1}).
\end{equation}
Thus, with $J=[(S^1)^2+(S^2)^2+(S^3)^2]^{1/2}$,
\begin{equation*}
\begin{split}
\D_\mu z\D^\mu z&=\sum_{j=1}^3(\partial_j z)^2+O(r^{-3})\\
&=\frac{1}{4M^4}\sum_{j=1}^3[S^jr^{-1}-x^jr^{-3}(S^1x^1+S^2x^2+S^3x^3)]^2+O(r^{-3})\\
&=\frac{1}{4M^4}[J^2r^{-2}-r^{-4}(S^1x^1+S^2x^2+S^3x^3)^2]+O(r^{-3}).
\end{split}
\end{equation*}
It follows that
\begin{equation}\label{e20}
\begin{split}
z^2+4M^2(y^2+z^2)\D_\mu z\D^\mu z&=\frac{1}{4M^4}(S^1x^1+S^2x^2+S^3x^3)^2r^{-2}+r^2\D_\mu z\D^\mu z+O(r^{-1})\\
&=\frac{J^2}{4M^4}+O(r^{-1}).
\end{split}
\end{equation}

\end{document}